\documentclass[journal,12pt,onecolumn,draftclsnofoot,]{IEEEtran} \def\OneCol

\IEEEoverridecommandlockouts

\usepackage{cite}
\usepackage{amsmath,amssymb,amsfonts}
\usepackage{tikz}
\usepackage{pgfplots}
\usetikzlibrary{arrows.meta}
\usepackage{amsmath}
\usepackage{amsfonts}
\usepackage{amssymb}
\usepackage{amsthm}
\usepackage{mathtools}
\usepackage{graphicx}
\usepackage{textcomp}
\usepackage{acronym}
\usepackage{algorithm}
\usepackage{algpseudocode}
\usepackage{setspace}
\newtheorem{lemma}{Lemma}

\newcommand\numberthis{\addtocounter{equation}{1}\tag{\theequation}}
\usepackage{xcolor}
\def\BibTeX{{\rm B\kern-.05em{\sc i\kern-.025em b}\kern-.08em
    T\kern-.1667em\lower.7ex\hbox{E}\kern-.125emX}}
    
\makeatletter
\newcommand*{\rom}[1]{\expandafter\@slowromancap\romannumeral #1@}
\makeatother

\acrodef{PLS}[PLS]{Physical-layer security}
\acrodef{DMC}[DMC]{discrete memoryless channel}
\acrodef{DMCs}[DMCs]{discrete memoryless channels}
\acrodef{BER}[BER]{bit error rate}
\acrodef{LDPC}[LDPC]{low-density parity-check codes}
\acrodef{LP}[LP]{linear programming}
\acrodef{RV}[RV]{random variable}
\acrodef{RVs}[RVs]{random variables}
\acrodef{IoT}[IoT]{Internet of things}
\acrodef{D2D}[D2D]{device-to-device}
\acrodef{BEC}[BEC]{binary erasure channel}
\acrodef{BEWC}[BEWC]{binary erasure wiretap channel}
\acrodef{BSC}[BSC]{binary symmetric channel} 
\acrodef{BSWC}[BSWC]{binary symmetric wiretap channel}
    
\begin{document}

\title{Secrecy Coding for the Binary Symmetric Wiretap Channel via Linear Programming}
%
%
%
\author{Ali~Nikkhah, \IEEEmembership{Student Member, IEEE}, Morteza~Shoushtari, \IEEEmembership{Student Member, IEEE}, Bahareh~Akhbari,~\IEEEmembership{Member,~IEEE}, and~Willie K. Harrison,~\IEEEmembership{Senior Member,~IEEE}

\thanks{A. Nikkhah was (at the time of conducting this research) and B. Akhbari is with the Faculty of Electrical Engineering, K. N. Toosi University of Technology, Tehran, 16317-14191, Iran, emails:\{\texttt{ali.nikkhah,akhbari@kntu.ac.ir}\}.\\
M. Shoushtari and W. K. Harrison are with the Department of Electrical and Computer Engineering, Brigham Young University, Provo, UT 84602 USA, emails:\{\texttt{morteza.shoushtari,willie.harrison@byu.edu}\}.\\
A. Nikkhah is currently with the Department of Computer and Information Science, Linköping University, Sweden, email:\{\texttt{ali.nikkhah@liu.se}\}.\\
The contributions of M. Shoushtari and W. K. Harrison were funded by US NSF grant \#1910812.}
}
\maketitle
\begin{abstract}
In this paper, we use a linear programming (LP) optimization approach to evaluate the equivocation for a wiretap channel where the main channel is noiseless, and the wiretap channel is a binary symmetric channel (BSC). Using this technique, we present an analytical limit for the achievable secrecy rate in the finite blocklength regime that is tighter than traditional fundamental limits. We also propose a secrecy coding technique that outperforms random binning codes. When there is one overhead bit, this coding technique is optimum and achieves the analytical limit. For cases with additional bits of overhead, our coding scheme can achieve equivocation rates close to the new limit. Furthermore, we evaluate the patterns of the generator matrix and the parity-check matrix for linear codes and we present binning techniques for both linear and non-linear codes using two different approaches: recursive and non-recursive. To our knowledge, this is the first optimization solution for secrecy coding obtained through linear programming.\\
\end{abstract}

\section{Introduction} \label{Introduction}
Wireless communications are inherently subjected to eavesdropping by unauthorized parties, which causes significant challenges to secure communications. To achieve confidentiality and authentication in wireless communications, traditional secrecy approaches use various cryptographic algorithms. These algorithms have been mainly implemented at higher layers of the network protocol stack.
Despite the fact that the upper layers of the network stack typically ensure secrecy by using key-based encryption algorithms, emerging new technologies such as the \ac{IoT} and~\ac{D2D} communications bring new limitations and challenges including complexity and power consumption, which make it difficult to use traditional cryptosystems in some specific environments. Motivated by these constraints and challenges, creating secure data transmission schemes based on the physical characteristics of wireless channels has garnered a lot of recent attention. 

\ac{PLS} was derived from information-theoretic definitions of secrecy that are rooted in Shannon's theories~\cite{Shannon1948,Shannon1949}. \ac{PLS} has some benefits over traditional cryptosystems, such as fast encoding/decoding procedures, and simple implementation, alongside providing quantifiable and provably unbreakable secrecy. The wiretap channel, which was first developed by Wyner in~\cite{Wyner1975} and later enhanced in~\cite{Csiszar1978, Cheong1978}, has been used to study and evaluate~\ac{PLS} over various communication scenarios. Wyner introduced the notion that secrecy can be provided by the characteristics of the communications channel itself without relying on shared secret keys. 

In the wiretap channel, as illustrated in Fig.~\ref{systemmodel}, there are three major participants: the transmitter (Alice), the legitimate receiver (Bob), and the eavesdropper (Eve). Alice intends to transmit a confidential message to Bob while keeping it secret from Eve. Additionally, Wyner in~\cite{Wyner1975} provided a coset coding scheme for the most basic version of the wiretap channel, where he assumed that the main channel (the channel between Alice and Bob) is noiseless, and the wiretap channel (the channel between Alice and Eve) is noisy. The coding scheme can ensure reliability for Bob and secrecy against Eve. 
This channel model is a specific instance of a more generally-defined degraded wiretap channel model~\cite{Csiszar1978}, wherein Eve observes a noisy version of Bob's observation. The idea of stochastic degradedness relaxes strict degradedness so that Eve's observation need only be noisier than Bob's observation. Many early secrecy results over the wiretap channel required at least stochastic degradedness, and thus Wyner's achievements went inactive since many practical communication scenarios did not satisfy the degraded model requirements. With the advent of wireless communications and mobile cellular networks, however,  this theoretical channel model has attracted more attention from both academia and industry~\cite{ShoushtariITC2021,Wang2019, shoushtariWTS2021, Shoushtari2022}.

The objectives of physical-layer security are to maximize the rate of reliable communication between Alice and Bob while simultaneously achieving secrecy against Eve. To achieve these objectives, Alice applies a \emph{secrecy code}, which converts the confidential message $M$ into an $n$-bit codeword $X^n$. 
The notion of \emph{perfect secrecy} was first introduced by Shannon in~\cite{Shannon1949}. He defined a communication scheme to achieve \emph{perfect secrecy}, if $I(M;X^n)=0$, meaning the codeword $X^n$ does not provide any information about the transmitted confidential message $M$. 
After Shannon defined perfect secrecy, he also proved it was unattainable in practice; therefore, Wyner later suggested the concept of \emph{weak secrecy}, which is achieved if the rate of the information leakage ${I(M;Z^n)}/{n}$ vanishes as ${n \to \infty}$, where $Z^n$ is the noisy observation of $X^n$ made by the eavesdropper through the wiretap channel.  Moreover, Maurer in \cite{Maurer1994}, suggested the notion of \emph{strong secrecy}, meaning $I(M;Z^n)$ approaches zero as $n \to \infty$. 

All of the aforementioned references in the previous paragraph considered the infinite blocklength regime, i.e., where $n$ approaches infinity. However, the effectiveness of secrecy codes in practical applications requires finite blocklength codes and secrecy metrics that make sense in the finite blocklength regime. This prompted the development of more empirical metrics, such as exact equivocation or~\ac{BER}~\cite{shoushtariITW2021, Shirvanimoghaddam2019, Shoushtari2023}. Using the~\ac{BER} as a security guarantee has its limitations, leaving the notion of equivocation as our metric of choice. The equivocation is defined as
\begin{equation}
\label{eq:Delta}
    \Delta = H(M|Z^n).
\end{equation}
It is necessary to understand the maximum possible equivocation under the finite blocklength constraint and to find efficient ways to calculate $\Delta$ for specific coding strategies.

\subsection{Prior work} \label{Coding for secrecy}
Many different secrecy codes have been developed for various versions of the wiretap channel~\cite{Harrison2013, Bloch2015}. The majority of secrecy codes are formed from a structure comparable to Wyner's coset coding approach, which as previously mentioned was first introduced by Wyner in the 1970s~\cite{Wyner1975}. For example, in \cite{Thangaraj2007} authors considered a~\ac{BEWC}, when the main channel is noiseless, and the wiretap channel is a~\ac{BEC}, and they presented a coding scheme based on~\ac{LDPC} that achieved the secrecy capacity by using the threshold properties of the codes. They also assumed a case where the main channel is noiseless and the wiretap channel is a~\ac{BSC} and they provided a coding scheme by using codes that have good error-detecting capabilities. Enhancements to~\ac{LDPC} coding schemes have resulted in stronger results over the~\ac{BEWC} as in~\cite{Subramanian2011}. In~\cite{Yang2016}, authors proposed a rateless coding scheme that also used polar codes to guarantee both security and reliability over the~\ac{BEWC}. Also, the performance of polar and Reed-Muller codes over the~\ac{BEWC} have been discussed in~\cite{Shakiba2021}, even in the case of finite blocklengths. 

The general idea of utilizing polar codes to achieve weak and strong secrecy for the~\ac{BSWC} with a degraded eavesdropper, has been presented in~\cite{Anderson2010,Hof2010,Mahdavifar2011, chen2010}. The vast majority of results in secrecy coding attempt to achieve either weak or strong secrecy, but the works of~\cite{Hassan2013} and~\cite{Pfister2017} began the search for best codes with finite blocklengths.
Reference~\cite{Hassan2013} addressed the channel model with a noiseless main channel and a binary symmetric wiretap channel, while~\cite{Pfister2017} opted for a binary erasure channel as the wiretap channel. Both approaches focused on exact equivocation calculations, and the techniques change with the choice of the channel model. Reference \cite{ranashort2023} further provides a new coding approach for the Gaussian wiretap channel. 

Recently, several additional works have focused on quantifying equivocation and finding the characteristics of the best short wiretap codes~\cite{Pfister2017, Harrison2018, Dong2020}, but only a small number of studies have been conducted on non-linear coding approaches. For example, in~\cite{Mohamed2017} authors presented a  scrambling module that combines linear and non-linear structures based on polar codes for the~\ac{BSWC}. Also, a non-linear cryptosystem based on structured~\ac{LDPC} codes with a small key size, which is able to provide secrecy and reliability at the same time, was proposed in~\cite{Stuart2017}.

There is a growing interest in comprehending the achievable rates and design principles of secrecy coding under finite blocklength constraints. Several notable contributions have made significant progress in this research area. In~\cite{Wei2019}, the authors established nonasymptotic achievability and converse bounds on the maximum secret communication rate over a wiretap channel. Their results demonstrated that these novel bounds are uniformly tighter than existing ones, leading to the tightest bounds on the second-order coding rate for both discrete memoryless and Gaussian wiretap channels. Furthermore, reference~\cite{Polyanskiy2010} investigated the maximal channel coding rate attainable at a given blocklength and error probability. They introduced new bounds on achievability and a converse that tightly encapsulates the fundamental limits, even for relatively short blocklengths. Another intriguing aspect explored in~\cite{Yang2017} is the maximum secrecy rate for a semi-deterministic wiretap channel, where the channel between the transmitter and the legitimate receiver is deterministic while that between the transmitter and the eavesdropper is a discrete memoryless channel.

The problem of identifying the best codes over various channel models remains unsolved, and although the gap between infinite blocklength achievable secrecy rates and achievable rates for finite codes has been addressed~\cite{Pfister2017}, precise limits of performance are still unknown for many finite blocklength codes.
These studies collectively contribute valuable insights into understanding and optimizing secrecy coding performance under finite blocklength constraints.

\subsection{Our Contributions} 
\label{Our contribution}
In this paper, we consider a binary symmetric wiretap channel model, where the main channel is noiseless, and the wiretap channel is a~\ac{BSC}. We evaluate and quantify the equivocation of the~\ac{BSWC} in the finite blocklength regime and obtain a new bound on achievable rates using~\ac{LP} techniques. Moreover, we present linear and non-linear coding schemes that achieve or approach this new fundamental limit. To the best of our knowledge, this is the first application of LP to achievable rates and coding over the wiretap channel. 

We summarize the main contribution of this work as follows:
\begin{enumerate}
    \item First, we develop the equivocation calculation for general binning techniques over the~\ac{BSWC}. Then, to optimize the equivocation, we formulate the calculation as a linear programming problem by defining the objective function and constraints. Our investigations lead to the discovery of a novel \textit{LP-derived limit} on the equivocation rate for finite blocklength through the solution of the~\ac{LP} problem. Remarkably, our numerical results show that our LP-derived limit improves on fundamental infinite blocklength performance bounds by incorporating a blocklength constraint. 
    \item We design a secrecy coding scheme and define its main components. Based on the results of the analysis, we demonstrate that our proposed coding scheme can achieve the new limit in the case where the size of the overhead bits $(n-k)$ is equal to one, and in other cases can get close to the new limit. In other words, our proposed coding scheme has better secrecy performance compared to a large number (perhaps all) of binning codes in the small blocklength regime.
    \item Finally, we provide recursive and non-recursive techniques to construct bins of the proposed secrecy coding.
\end{enumerate}

\subsection{Organization of the Paper} 
\label{Organization}
The rest of the paper is organized as follows. In Section \ref{System Model and Preliminaries}, we review the system model and the bin coding that we consider in this work and the basic coset coding technique. 
Section \ref{Equivocation calculation} explores the equivocation for linear block codes and evaluates a new search method for non-linear codes. In Section \ref{Optimization}, we describe the utilization of the~\ac{LP} approach to find new achievable limits. Section \ref{NA} presents our proposed coding schemes (Ni codes). In Section \ref{The Answer of optimization for different cases}, we explore the utilization of the~\ac{LP} approach for some numerical examples and then we present the results. In Section \ref{Matrices to generate the code}, we explain the construction of the generator and the parity-check matrices for our coding scheme by showing patterns and numerical examples. In Section \ref{The search for $AC+B$}, we unify our two coding approaches and present an integrated process to produce the bins. Section \ref{Conclusion}, provides concluding remarks and ideas for future work.

\subsection{Notation} \label{Notation}
Throughout this paper, all random variables (RVs) are denoted by capital letters (e.g., $A$), while small letters (e.g., $a$) represent their corresponding realizations. The calligraphic forms (e.g., $\mathcal{A}$) denote sets and $|\cdot|$ denotes the cardinality of a set (e.g., $|\mathcal{A}|$). 
Capital and small bold letters (e.g., $\mathbf{A}$ and $\mathbf{a}$) denote matrices and vectors, respectively, and superscripts indicate the length of vectors when used.  All vectors are row vectors, and all codes are binary. 
All logarithms are of base two.

\section{System Model and Preliminaries}  
\label{System Model and Preliminaries}
In this work, before diving into the details of our novel non-linear secrecy code, its construction, and the assessment of its secrecy performance, we first lay the groundwork by explaining the system model under consideration. Following this, we present our binning encoder and decoder and offer a brief overview of linear coset codes' encoding and decoding processes.
\subsection{System Model} \label{System Model}
Let us consider the wiretap channel system model depicted in Fig. \ref{systemmodel}, where Alice intends to transmit the confidential message $M$, which is assumed to be chosen uniformly at random from $\mathcal{M}=\{1,2,...,2^k\}$, to Bob through the main channel, and wishes to keep it secret from Eve, who listens to the message via her own channel. Both the main and wiretap channels are discrete memoryless channels, and in particular, the wiretap channel is stochastically degraded with respect to the main channel. To keep the message secure, Alice encodes $M$ into an $n$-bit binary codeword $X^n$. Bob receives $Y^n$ through the main channel and uses a decoding function to map $Y^n$ into the estimated version of the confidential message $\hat{M}$.
Eve observes $Z^n$ through the wiretap channel.
In this work, the main channel is noiseless, meaning $Y^n=X^n$, and the wiretap channel is a ~\ac{BSC}, with crossover probability $p$.

\ifdefined\OneCol
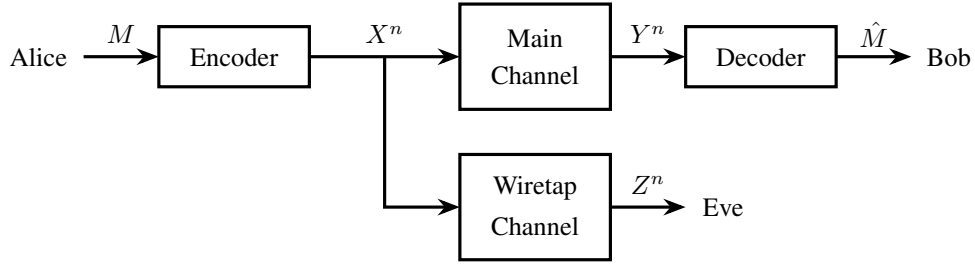
\begin{figure}
  \centering
    \resizebox{0.8\textwidth}{!}{%
\begin{tikzpicture}
\draw[-Stealth, very thick] (0,0) -- (1,0) node[xshift=-1.6cm] {\small Alice} node[xshift=-0.5cm, yshift=0.3cm] {\small $M$};
\draw[very thick] (1,-0.4) rectangle (3,0.4) node[xshift=-1cm,yshift=-0.4cm] {\small Encoder};
\draw[-Stealth, very thick] (3,0) -- (5,0) node[xshift=-1cm, yshift=0.3cm] {\small $X^n$};
\draw[-Stealth, very thick] (4,0) -- (4,-2) -- (5,-2);
\draw[very thick] (5,-0.7) rectangle (7,0.7) node[xshift=-1cm,yshift=-0.45cm] {\small Main} node[xshift=-1cm,yshift=-0.95cm] {\small Channel} ;
\draw[very thick] (5,-0.7-2) rectangle (7,0.7-2) node[xshift=-1cm,yshift=-0.45cm] {\small Wiretap} node[xshift=-1cm,yshift=-0.95cm] {\small Channel} ;
\draw[-Stealth, very thick] (7,0)-- (8,0) node[xshift=-0.5cm, yshift=0.3cm] {\small $Y^n$};
\draw[very thick] (1+7,-0.4) rectangle (3+7,0.4) node[xshift=-1cm,yshift=-0.4cm] {\small Decoder}  ;
\draw[-Stealth, very thick] (7,-2)-- (8,-2) node[xshift=0.5cm] {\small Eve} node[xshift=-0.5cm, yshift=0.3cm] {\small $Z^n$};
\draw[-Stealth, very thick] (10,0)-- (11,0) node[xshift=0.5cm] {\small Bob} node[xshift=-0.5cm,yshift=0.3cm] {\small $\hat{M}$};
\end{tikzpicture}
}
  \caption{Wiretap channel model.}
  \label{systemmodel}
\end{figure}
\fi

\ifdefined\TwoCol
\begin{figure}
  \centering
  \resizebox{0.49\textwidth}{!}{%
\begin{tikzpicture}
\draw[-Stealth, very thick] (0,0) -- (1,0) node[xshift=-1.6cm] {\small Alice} node[xshift=-0.5cm, yshift=0.3cm] {\small $M$};
\draw[very thick] (1,-0.4) rectangle (3,0.4) node[xshift=-1cm,yshift=-0.4cm] {\small Encoder};
\draw[-Stealth, very thick] (3,0) -- (5,0) node[xshift=-1cm, yshift=0.3cm] {\small $X^n$};
\draw[-Stealth, very thick] (4,0) -- (4,-2) -- (5,-2);
\draw[very thick] (5,-0.7) rectangle (7,0.7) node[xshift=-1cm,yshift=-0.45cm] {\small Main} node[xshift=-1cm,yshift=-0.95cm] {\small Channel} ;
\draw[very thick] (5,-0.7-2) rectangle (7,0.7-2) node[xshift=-1cm,yshift=-0.45cm] {\small Wiretap} node[xshift=-1cm,yshift=-0.95cm] {\small Channel} ;
\draw[-Stealth, very thick] (7,0)-- (8,0) node[xshift=-0.5cm, yshift=0.3cm] {\small $Y^n$};
\draw[very thick] (1+7,-0.4) rectangle (3+7,0.4) node[xshift=-1cm,yshift=-0.4cm] {\small Decoder}  ;
\draw[-Stealth, very thick] (7,-2)-- (8,-2) node[xshift=0.5cm] {\small Eve} node[xshift=-0.5cm, yshift=0.3cm] {\small $Z^n$};
\draw[-Stealth, very thick] (10,0)-- (11,0) node[xshift=0.5cm] {\small Bob} node[xshift=-0.5cm,yshift=0.3cm] {\small $\hat{M}$};
\end{tikzpicture}
}
  \caption{Wiretap channel model.}
  \label{systemmodel}
\end{figure}
\fi






The design of the coding scheme in the wiretap channel has two objectives: the first one is to maximize the information rate of reliable communication between Alice and Bob, which is called the \emph{reliability constraint}; and the second objective is to ensure that Eve gets as little information as possible about $M$, which is called the \emph{secrecy constraint}. The reliability constraint is measured in terms of the average probability of error
\begin{equation} \label{reliability}
P_e=\frac{1}{|\mathcal{M}|} \sum\limits_{m \in \mathcal{M}} \Pr(m \neq \hat{m}|m).
\end{equation}
Since we assume the main channel is noiseless, $P_e$ is equal to zero. Assuming a noiseless main channel allows us to focus only on secrecy in the optimization, although future efforts consider an error-prone main channel. A similar approach is taken in other secrecy coding works \cite{Harrison2013}, \cite{Mahdavifar2011}, \cite{Pfister2017}, \cite{Suresh2010}, \cite{Harrison2018} and \cite{Jensen2019}.

The equivocation on the eavesdropper's side is used to measure the secrecy constraint and is defined as $\Delta$ in (\ref{eq:Delta}).
The equivocation rate is $R_e=\frac{1}{n}\Delta$.

\subsection{Preliminaries}
In a wiretap channel with a noiseless main channel, the objective is to formulate optimal codes for any given $n$ and $k$ such that $\Delta \rightarrow H(M)$. It is important to highlight that evaluating the security performance for small to moderate blocklengths poses significant challenges.
The codes discussed in this paper predominantly fall under the non-linear category. As a result, we typically employ the term \textit{bin} instead of \textit{coset}. However, when the term \textit{coset} is used, it specifically refers to a linear code.

The encoder and decoder proposed in this paper can be described as follows.
Given a message selected from the set $\mathcal{M} = \{1,2,...,2^k\}$, one can utilize a fundamental binning encoder, which isn't strictly derived from linear codes. Specifically, think of the sets of codewords as bins: $\mathcal{B}_1, \mathcal{B}_2, \ldots, \mathcal{B}_{2^k}$. Each set of codewords, denoted $\mathcal{B}_i$, comprises $2^l$ codewords linked to the message $m_i$, where $i$ ranges from 1 to $2^k$. Using a stochastic encoding approach, the message is encoded into a random codeword from the associated set. When decoding, the received codeword is examined against the $2^k$ bins. Should the codeword correspond to the $i$th bin, $\mathcal{B}_i$, the deterministic decoder identifies it with the $i$th message, $m_i$.

Nevertheless, some of our codes are linear, and we will present the generator and parity-check matrices for them. Hence, we would have a review of the wiretap coset coding.

The fundamental concept behind Wyner's coding scheme~\cite{Wyner1975}, is to inject randomness into the encoding process. This randomness adds confusion to the eavesdropper's observation by providing multiple codewords for each message, just as in the general binning scheme mentioned above.
Let $\mathcal{C}$ be an $(n,k)$ linear block code, with cosets $\mathcal{C}_1,\mathcal{C}_2,\ldots,\mathcal{C}_{2^k}$, where $\mathcal{C}_1 = \mathcal{C}$ and $l=n-k$. Also, let $\mathbf{G}'$ be the $(l)\times n$ generator matrix of $\mathcal{C}$ and $\mathbf{H}'$ be the $k \times n$ parity-check matrix of $\mathcal{C}$. The encoder function encodes the message $m$ by choosing a codeword from $2^l$ different codewords, uniformly at random. This operation can be done using an auxiliary message $M' \in \mathbb{F}_2^{(n-k)}$, where $M'$ carries no real information and is chosen uniformly at random. Then $\mathbf{G}$, which is created by vertically combining $\mathbf{G}''$ and $\mathbf{G}'$, defined by 
\begin{equation}
  \mathbf{G} = \left[ \begin{matrix} \mathbf{G}'' \\ \mathbf{G}' \end{matrix} \right],
\end{equation}
is required for the encoding process. The matrix $\mathbf{G}''$ is comprised of $k$ linearly independent vectors in $\mathbb{F}_2^{n}$, but not in  $\mathcal{C}$. Using this construction approach, $\mathbf{G}$ will be a square $(n \times n)$ full-rank matrix. Then codeword $X^n$ is calculated as follows
\begin{equation}
  \mathbf{x}^n = \left[ \begin{matrix} \mathbf{m} & \mathbf{m}' \end{matrix}\right] \left[ \begin{matrix} \mathbf{G}'' \\ \mathbf{G}' \end{matrix} \right] = \mathbf{m} \mathbf{G}'' + \mathbf{m}'\mathbf{G}'.
\end{equation}
Clearly, $\mathbf{m}\mathbf{G}''$ chooses the coset, and $\mathbf{m}'\mathbf{G}'$ chooses a specific codeword from the coset uniformly at random.

Since we assumed that the main channel is noiseless on this \ac{BSWC} model, Bob can perform the decoding process by simply mapping the codeword back to the appropriate message. First, the decoder function computes a syndrome as follows
\ifdefined\OneCol
\begin{equation}
\mathbf{\boldsymbol{\sigma}} = \mathbf{y}^n(\mathbf{H}')^T = \mathbf{x}^n(\mathbf{H}')^T = \mathbf{m} \mathbf{G}''(\mathbf{H}')^T + \mathbf{m}'\mathbf{G}'(\mathbf{H}')^T = \mathbf{m} \mathbf{G}''(\mathbf{H}')^T.
\end{equation}
\fi
\ifdefined\TwoCol
\begin{align*}
\mathbf{\boldsymbol{\sigma}} &= \mathbf{y}^n(\mathbf{H}')^T = \mathbf{x}^n(\mathbf{H}')^T = \mathbf{m} \mathbf{G}''(\mathbf{H}')^T + \mathbf{m}'\mathbf{G}'(\mathbf{H}')^T\\
&= \mathbf{m} \mathbf{G}''(\mathbf{H}')^T. \numberthis 
\end{align*}
\fi
It is worth mentioning that $\mathbf{G}''(\mathbf{H}')^T$ forms a bijective mapping between $\boldsymbol{\sigma}$ and $\mathbf{m}$. If $\mathbf{G}''$ and $\mathbf{H}'$ are chosen so that $\mathbf{G}''(\mathbf{H}')^T$, is equal to the $(k \times k)$ identity matrix, then $\boldsymbol{\sigma} = \mathbf{m}$. Otherwise, the mapping must be inverted to complete the decoder~\cite{Pfister2017}.

\section{Equivocation Calculation} \label{Equivocation calculation}
We take our equivocation calculation over the~\ac{BSWC} as an extension to the process described in~\cite{Pfister2018}.  
Therein, the author calculated the equivocation for finite blocklength linear codes over the~\ac{BSWC}. In this paper, we expand on this analysis to include non-linear codes as well as linear codes. 

Consider $X^n$ transmitted over a BSC($p$) to produce $Z^n$. Then set
\begin{equation} \label{Gamma_f}
\Gamma (\eta)= P(\text{a specific $\eta$-error pattern})= p^{\eta} \times q^{n-\eta},
\end{equation}
where $q=1-p$. Thus, the probability of $x^n$ being transmitted, given that Eve observes $z^n$ through her channel, is equal to
\begin{equation} \label{Hamming_dis_pr}
P(x^n|z^n)=\Gamma (d_H (x^n,z^n)),
\end{equation}
where $d_H (.,.)$ signifies the Hamming distance between the two parameters.
Furthermore, the probability that $z^n$ belongs to the bin $\mathcal{B}_i$, which corresponds to the confidential message $m_i$, is defined as
\begin{equation} \label{Coset_pr}
P_{\mathcal{B}_i}=P(\mathcal{B}_i|z^n)=\sum_{x^n \in \mathcal{B}_i} P(x^n|z^n).
\end{equation}
The equivocation of the eavesdropper is then defined as
\begin{equation}  \label{Conditioned_Ent}
H(M|Z^n)=- \sum_{i=1}^{2^k} P_{\mathcal{B}_i} \text{log} P_{\mathcal{B}_i}. 
\end{equation}
Because $M$, $X^n$, and the $Z^n$, are distributed uniformly, the average equivocation is
\begin{equation} \label{equivovation_t} 
H(M|Z^n)=\frac{1}{2^n}\sum\limits_{z^n\in Z^n}H(M|Z^n=z^n). \numberthis 
\end{equation}

From another perspective, there can be a $2^k \times (n+1)$ matrix $\mathbf{R}$ for a specific observation $z^n$ that illustrates the relationship of $z^n$ with each codeword in the bins. The element in position $(i,j)$ in the matrix $\mathbf{R}$, gives the number of codewords in bin $\mathcal{B}_i$, that have Hamming distance of $j-1$ with $z^n$. 
Now consider an $(n+1)\times 1$ matrix $\mathbf{\Gamma}$, where each element of the matrix $\mathbf{\Gamma}$ ($\gamma_{i}$) equals $\Gamma(i-1)$.

The matrix $\mathbf{\Pi}$ is then written as follows
\begin{equation}
\label{RGAMA}
\mathbf{\Pi}=\mathbf{R}\mathbf{\Gamma},
\end{equation}
and $\mathbf{\Pi}^{T}=[P_{\mathcal{B}_1}, P_{\mathcal{B}_2}, \hdots,P_{\mathcal{B}_{2^k}}]$. For linear wiretap codes, if a ${z^{'}}^n$ is added to the codewords of the bins, the ordering of the bins will change, but the grouping of codewords in bins will remain the same. Thus, the entropy of the matrix $\mathbf{\Pi}$ will remain the same. Hence, the total equivocation is the same as a conditioned equivocation for each of the eavesdropper's observations. 

Whereas, for the non-linear coding scheme, each $z^n$ can have a different $\mathbf{R}$, so each $\mathbf{\Pi}$ calculation can result in a different entropy. The ultimate goal is to maximize total equivocation in (\ref{equivovation_t}), which equals the average entropy of all $2^n$ different $z^n$ values. Our first contribution is to evaluate how a single entropy of an arbitrary $z^n$ can be maximized, in general. In other words, how should $\mathbf{R}$ be designed to maximize the entropy obtained from the corresponding $\mathbf{\Pi}$. 
The difficult part is choosing $2^k$ rows for $\mathbf{R}$ from all the possible rows, resulting in maximization over $\mathbf{\Pi}$.
We consider a possible row as a vector of $(n+1)$ non-negative integers when the summation of the number of codewords in each bin is $e=2^l$. Thus, the number of different possible rows is equal to
\begin{equation} \label{simplenumber}
N=\binom{e+n}{e}.
\end{equation}
Another equivalent calculation of $N$ is presented in Appendix A. Notice that all three formulas of (\ref{simplenumber}), (\ref{appnumber}), and (\ref{appnumber2}) result in the same answer.
\section{Optimization}
\label{Optimization}
In this section, we define a method for selecting $2^k$ possible rows of $\mathbf{R}$ to maximize $\mathbf{\Pi}$ based on linear programming (LP). The \ac{LP} optimization problem can be formulated as  
\begin{align*}  
\text{maximize} \ \mathbf{f}^\intercal \mathbf{x} \ & \text{with the respect to} \ \mathbf{x}\\  
\ & \text{s.t.} \ \mathbf{Ax} = \mathbf{b}
\\  & \text{and} \ \mathbf{0} \leq \mathbf{x}  . \numberthis \label{optimization}
\end{align*}


In this optimization problem, we are looking to maximize the objective function $\mathbf{f}^\intercal \mathbf{x}$, over the feasible set.
The ${N \times 1}$ vector $\mathbf{f}$ contains the values of $- P_{\mathcal{B}_i} \text{log}(P_{\mathcal{B}_i})$, when $P_{\mathcal{B}_i}=\mathbf{r}_i \mathbf{\Gamma}$, in which $\mathbf{r}_i \in (\mathbb{Z}^+)^{n+1}$ is the $i$th possible row of the total $N$ rows of $\mathbf{R}$. The possible rows are arranged in an arbitrary manner, and this arrangement holds throughout the entire problem. The ${N \times 1}$ vector $\mathbf{x}$ is the number of times each possible row could be selected. The $i$th column of the matrix $\mathbf{A}_{(n+1) \times N}$ corresponds to the $i$th possible row, and the $i$th element of $\mathbf{b}_{(n+1) \times 1}$ is equal to $\binom{n}{i-1}$.
The equality constraint in (\ref{optimization}) ensures that the summation of the elements in the $i$th column of $\mathbf{R}$ is, as it should be, $\binom{n}{i-1}$. The lower bound of $\mathbf{0}_{N \times 1}$ indicates that no negative quantity can exist in each row. Given the equality constraint, an upper bound is not required.  

The objective function in (\ref{optimization}) is linear and its constraint forms a polyhedron. Hence, this is an LP problem; thus, convex \cite{Boyd2004}, and the optimum answer will occur at one of the polyhedron's corners, which will be in integer values.
By solving this LP problem, for any form of $(l,k)$ and also for a specific $p$, we are able to obtain $\mathbf{x^*}$, which is the answer of the optimization in (\ref{optimization}) that indicates how many potential rows should be selected to maximize the entropy in (\ref{Conditioned_Ent}).

\begin{lemma}
\label{lem:1}
For a specific form of $(l,k)$ and a given $p$, the value of optimum entropy for a conditioned equivocation, i.e., (\ref{Conditioned_Ent}), is an upper bound for the total equivocation, i.e., (\ref{equivovation_t}).
\end{lemma}

\begin{proof}
As proved by the optimization problem, this optimal value is the maximum amount of message entropy conditioning a $z^n$. In order for the total equivocation to be maximized, the distribution of codewords in the code table should be such that all $2^n$ conditioned equivocations have the optimum value. So, by (\ref{equivovation_t}), the total equivocation would have the optimum value regardless of whether the code table has such a codeword distribution or not. This is an upper bound on the total equivocation for a specific form of $(l,k)$ and a specific value of $p$.
\end{proof}
Consider the upper bound on equivocation derived from the technique described in Lemma~\ref{lem:1}, and compare the result to that obtained by Wyner's secrecy capacity~\cite{Wyner1975}. Note that our result provides a tighter bound than Wyner, as we consider a finite blocklength constraint in our optimization. In the remainder of the paper, we refer to our upper bound as the LP-derived limit. we refer to Wyner's bound as the infinite blocklength limit. Naturally, the LP-derived limit approaches the finite blocklength limit as $n\to\infty$.


It should be noted that when an optimal limit is asserted based on the reasoning presented, it is a confirmed analytical limit for a form when all the $2^n$ codewords are distributed equally into the $2^k$ bins. In other words, this LP-derived limit is not a suggested limit for other structures that employ non-uniform distributions in the code table or when all $2^n$ codewords are not used during the encoding process. However, in this wiretap channel model, when the main channel is noiseless, the distribution is uniform, and all $2^n$ codewords will be used equally. The solutions to this optimization problem for different forms are be discussed in Section~\ref{The Answer of optimization for different cases}.
\section{Ni Coding and Its Two Essential Parts} \label{NA}
In this section, we define and explain the \emph{Ni} coding technique and its two main construction algorithms, to achieve the maximum equivocation in the~\ac{BSWC} channel model introduced in Section~\ref{System Model}. In the Ni code, \emph{forms} are created by different combinations of the length of overhead bits and information bits, $l$ and $k$, respectively. A specific form will be denoted by $(l, k)$. Furthermore, a \emph{case} contains all forms with the same number of overhead bits $(l)$. The first form (FF) of a case refers to the form with $k=1$.
Utilizing the Ni coding approach can produce all codewords for each form of $(l,k)$. The variable $\mathcal{D}$ represents a code table that the Ni code produces, and $\mathcal{D}_{l,k}$ is used to specify a table for a specific $(l, k)$ form with $2^k$ bins and $2^l$ codewords per bin. Columns of the table $\mathcal{D}_{l,k}$ are the bins $\mathcal{B}_1,\mathcal{B}_2,\ldots,\mathcal{B}_{2^k}$. Our proposed Ni code has the advantage of allowing $\mathcal{D}$ of any different form $(l,k)$ to be obtained from $\mathcal{D}_{l-1,k}$ or $\mathcal{D}_{l,k-1}$. As a result, more complex and higher forms can be generated from simpler ones. In Section \ref{The Answer of optimization for different cases}, we explain and explore in detail the secrecy performance of our coding scheme, and we demonstrate that they outperform a myriad of randomly generated codes in terms of equivocation rate. Next, we will talk about interconnections of different forms and cases. 

The recursive alternate single bit adding (RASBA) algorithm, which will be explained in more detail in Section \ref{RASBA}, allows the system to keep the number of overhead bits ($l$) constant while generating $\mathcal{D}$ of the forms with a higher value for $k$. To do so, as illustrated by the horizontal lines in Fig. \ref{Schememap}, moving on the green and blue lines increases the value of $n$ and $k$, while keeping the number of overhead bits constant. It should be noted that the green line represents when the system does not use any coding algorithm, i.e., $l=0$.

The recursive alternate half-bits adding (RAHBA) algorithm will be detailed in Section~\ref{RAHBA}. This algorithm enables the system to produce any $\mathcal{D}$ of the forms with higher $l$ and $n$ values while fixing the number of data bits at $k$. The red and the (solid) black lines in Fig. \ref{Schememap} can be used by the RAHBA approach to increase $l$, and consequently $n$, while keeping $k$ constant. It is worth mentioning that the first vertical line is separated by its color (red) from the other vertical lines (black) because by moving on the red line, the $\mathcal{D}$ of the FF of the higher cases can be obtained by the $\mathcal{D}$ of the FF of the previous cases. This is a part of the standard path described in Section \ref{The search for $AC+B$}, where the proposed method is formulated.

Moving along each of the two groups of lines causes an increment in $n$. Each intersection of the lines is a form, and the intersections accentuated with red dots are the FF of the various cases. All forms belonging to the same horizontal line are of the same case. 

\begin{figure}[ht]
\centerline{\begin{tikzpicture}

\node [xshift=-1.5cm, yshift=3.3cm,rotate=90] {To higher overhead bits by RAHBA $\rightarrow$};
\node [xshift=2.9cm, yshift=-1.5cm] {To higher data bits by RASBA $\rightarrow$};

\draw[-Stealth,red, very thick] (0,0) -- (0,7); \node [xshift=0cm, yshift=-0.6cm, rotate=90] {$k=1$};

\draw[densely dotted, very thick] (1,0) -- (1,1) node [xshift=0cm, yshift=-1.6cm, rotate=90] {$k=2$};
\draw[-Stealth, very thick] (1,1) -- (1,7);

\draw[densely dotted, very thick] (2,0) -- (2,1) node [xshift=0cm, yshift=-1.6cm, rotate=90] {$k=3$};
\draw[-Stealth, very thick] (2,1) -- (2,7);

\draw[densely dotted, very thick] (3,0) -- (3,1) node [xshift=0cm, yshift=-1.6cm, rotate=90] {$k=4$};
\draw[-Stealth, very thick] (3,1) -- (3,7);

\draw[densely dotted, very thick] (4,0) -- (4,1) node [xshift=0cm, yshift=-1.6cm, rotate=90] {$k=5$};
\draw[-Stealth, very thick] (4,1) -- (4,7);

\draw[densely dotted, very thick] (5,0) -- (5,1) node [xshift=0cm, yshift=-1.6cm] {$\hdots$};
\draw[-Stealth, very thick] (5,1) -- (5,7);

\draw[-Stealth,green, very thick] (0,0) -- (6,0); \node [xshift=-0.6cm, yshift=0cm] {$l=0$};

\draw[-Stealth,blue, very thick] (0,1) -- (6,1);
\node [xshift=-0.6cm, yshift=1cm] {$l=1$};

\draw[-Stealth,blue, very thick] (0,2) -- (6,2);
\node [xshift=-0.6cm, yshift=2cm] {$l=2$};

\draw[-Stealth,blue, very thick] (0,3) -- (6,3);
\node [xshift=-0.6cm, yshift=3cm] {$l=3$};

\draw[-Stealth,blue, very thick] (0,4) -- (6,4);
\node [xshift=-0.6cm, yshift=4cm] {$l=4$};

\draw[-Stealth,blue, very thick] (0,5) -- (6,5);
\node[xshift=-0.6cm, yshift=5cm] {$l=5$};

\draw[-Stealth,blue, very thick] (0,6) -- (6,6); \node [xshift=-0.6cm, yshift=6cm] {$\vdots$};
\draw[fill,red] (0,0) circle (0.08cm);
\draw[fill,red] (0,1) circle (0.08cm);
\draw[fill,red] (0,2) circle (0.08cm);
\draw[fill,red] (0,3) circle (0.08cm);
\draw[fill,red] (0,4) circle (0.08cm);
\draw[fill,red] (0,5) circle (0.08cm);
\draw[fill,red] (0,6) circle (0.08cm);

\end{tikzpicture}}
\caption{Scheme map of our Ni coding scheme.}
\label{Schememap}
\end{figure}
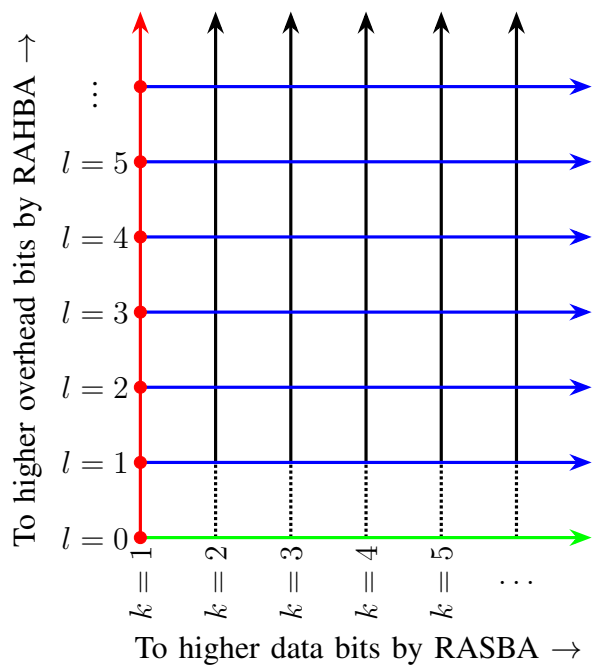

\subsection{RASBA Algorithm} \label{RASBA}
\begin{algorithm}[H]
\setstretch{1.2}
\caption{Constructing $\mathcal{D}_{l,k+1}$ given any $\mathcal{D}_{l,k}$}
\label{alg:rasba}
\begin{algorithmic}
\Require $\mathcal{D}_{l,k}$
\For{$i \gets 1$ to $2^k$}
    \State $\mathbf{X}_{i} \gets \mathcal{B}_i$ \texttt of $\mathcal{D}_{l,k}$
      \State $\mathbf{Y}_{2^l \times 1}, \mathbf{Z}_{2^l \times 1} \gets \emptyset$ \Comment{generate $2$ empty column vectors}
      \For{$j \gets 1$ to $2^l$}
      \If{$j$ is odd}
        \State $\mathbf{y}_{j} \gets \mathbf{x}_{j} + 0$ \Comment{appending 0 to the right side of the codeword $\mathbf{x}_j$}
        \State $\mathbf{z}_{j} \gets \mathbf{x}_{j} + 1$ \Comment{appending 1 to the left side of the codeword $\mathbf{x}_j$}
        \Else
        \State $\mathbf{y}_{j} \gets \mathbf{x}_{j} + 1$ 
        \State $\mathbf{z}_{j} \gets \mathbf{x}_{j} + 0$ 
        \EndIf
    \EndFor
    \State $\mathcal{B}_{2i-1}$ \texttt of $\mathcal{D}_{l,k+1} \gets \mathbf{Y}$ \Comment{all codewords in $\mathbf{Y}$ will produce bin $\mathcal{B}_{2i-1}$ in form $(l,k+1)$}
    \State $\mathcal{B}_{2i}$ \texttt of $\mathcal{D}_{l,k+1} \gets \mathbf{Z}$ \Comment{all codewords in $\mathbf{Z}$ will produce bin $\mathcal{B}_{2i}$ in form $(l,k+1)$}
\EndFor
\end{algorithmic}
\end{algorithm}
In this coding algorithm, the coding system will keep the number of overhead bits ($l$) constant and move to forms with higher $k$. Obviously, $n$ increases with $k$. It should be noted that the increment on $k$ is limited to a maximum of $1$ in each step. Since $l$ is constant, the number of codewords in each bin will not change with each move to the next form. 

To calculate $\mathcal{D}_{l,k+1}$ from the current  $\mathcal{D}_{l,k}$ by using RASBA algorithm, each bin $\mathcal{B}_1,\mathcal{B}_2,\ldots,\mathcal{B}_{2^k}$ of the current form $(l, k)$ generates two bins of the next form $(l,k+1)$. Therefore, each increment doubles the number of bins. The bin $\mathcal{B}_{2i-1}$ of the $\mathcal{D}_{l,k+1}$ will be produced by appending $0$ and $1$ to the right side of each codeword in the odd and even positions of the bin $\mathcal{B}_i$ of the current form $(l,k)$, respectively. Also, for the bin $\mathcal{B}_{2i}$ of the next form $(l,k+1)$, $1$ will be appended to codewords in the odd position, while $0$ will be appended to the codewords in the even position of the $\mathcal{B}_i$ in the current form $\mathcal{D}_{l,k}$. This method will generate all codewords in the next bin $\mathcal{B}_{2i}$ of the form $(l,k+1)$. These procedures will be repeated for all $i$ when $1 \leq i \leq 2^k$ to construct all codewords for all bins of the new form. The RASBA algorithm is described in Algorithm~\ref{alg:rasba}.
In practice, this coding algorithm allows the system to generate $\mathcal{D}_{l,k}$ from the previous $\mathcal{D}_{l,k-1}$, as shown in Fig. \ref{RASBA_Growth}.
\begin{figure}[ht]
\centering
\begin{tikzpicture}
\node [xshift=0.5cm, yshift=5cm] {$\mathcal{D}_{1,1}$};
\draw[-Stealth] (1,5) -- (1.5,5);
\node [xshift=2cm, yshift=5cm] {$\mathcal{D}_{1,2}$};
\draw[-Stealth] (2.5,5) -- (3,5);
\node [xshift=3.5cm, yshift=5cm] {$\mathcal{D}_{1,3}$};
\draw[-Stealth] (4,5) -- (4.5,5);
\node [xshift=5cm, yshift=5cm] {$\mathcal{D}_{1,4}$};
\draw[-Stealth] (5.5,5) -- (6,5);
\node [xshift=6.5cm, yshift=5cm] {$\mathcal{D}_{1,5}$};
\draw[-Stealth] (7,5) -- (7.5,5);
\node [xshift=8cm, yshift=5cm]
{$\dots$};

\node [xshift=0.5cm, yshift=4.5cm] {$\mathcal{D}_{2,1}$};
\draw[-Stealth] (1,4.5) -- (1.5,4.5);
\node [xshift=2cm, yshift=4.5cm] {$\mathcal{D}_{2,2}$};
\draw[-Stealth] (2.5,4.5) -- (3,4.5);
\node [xshift=3.5cm, yshift=4.5cm] {$\mathcal{D}_{2,3}$};
\draw[-Stealth] (4,4.5) -- (4.5,4.5);
\node [xshift=5cm, yshift=4.5cm] {$\mathcal{D}_{2,4}$};
\draw[-Stealth] (5.5,4.5) -- (6,4.5);
\node [xshift=6.5cm, yshift=4.5cm] {$\mathcal{D}_{2,5}$};
\draw[-Stealth] (7,4.5) -- (7.5,4.5);
\node [xshift=8cm, yshift=4.5cm]
{$\dots$};

\node [xshift=0.5cm, yshift=4cm] {$\mathcal{D}_{3,1}$};
\draw[-Stealth] (1,4) -- (1.5,4);
\node [xshift=2cm, yshift=4cm] {$\mathcal{D}_{3,2}$};
\draw[-Stealth] (2.5,4) -- (3,4);
\node [xshift=3.5cm, yshift=4cm] {$\mathcal{D}_{3,3}$};
\draw[-Stealth] (4,4) -- (4.5,4);
\node [xshift=5cm, yshift=4cm] {$\mathcal{D}_{3,4}$};
\draw[-Stealth] (5.5,4) -- (6,4);
\node [xshift=6.5cm, yshift=4cm] {$\mathcal{D}_{3,5}$};
\draw[-Stealth] (7,4) -- (7.5,4);
\node [xshift=8cm, yshift=4cm]
{$\dots$};
\node [xshift=3.75cm, yshift=3.5cm]
{$\vdots$};
\end{tikzpicture}
\caption{Generating different forms by RASBA method.}
  \label{RASBA_Growth}
\end{figure}
At the end of this section, an illustration of the transition of some of these $\mathcal{D}_{l,k}$ from their previous forms are shown.  
\subsection{RAHBA Algorithm} \label{RAHBA}
\begin{algorithm}[H]
\setstretch{1.2}
\caption{Constructing $\mathcal{D}_{l+1,k}$ given any $\mathcal{D}_{l,k}$}
\label{alg:rahba}
\begin{algorithmic}
\Require $\mathcal{D}_{l,k}$
\For{$i \gets 1$ to $\frac{2^k}{2}$}
    \State $\mathbf{X}_{i} \gets \mathcal{B}_{2i-1}$ \texttt of $\mathcal{D}_{l,k}$
    \State $\mathbf{Y}_{i} \gets \mathcal{B}_{2i}$ \texttt of $\mathcal{D}_{l,k}$
        \State $\mathbf{V}_{2^l \times 1}, \mathbf{W}_{2^l \times 1}, \mathbf{U}_{2^l \times 1}, \mathbf{Z}_{2^l \times 1} \gets \emptyset$ \Comment{generate $4$ empty column vectors}
      \For{$j \gets 1$ to $2^l$}
      \State $\mathbf{v}_{j} \gets \mathbf{x}_{j} + 0$ \Comment{appending 0 to the right side of the codeword $\mathbf{x}_j$}
      \State $\mathbf{w}_{j} \gets \mathbf{x}_{j} + 1$ \Comment{appending 1 to the right side of the codeword $\mathbf{x}_j$}
      \State $\mathbf{u}_{j} \gets \mathbf{y}_{j} + 0$ \Comment{appending 0 to the right side of the codeword $\mathbf{y}_j$}
      \State $\mathbf{z}_{j} \gets \mathbf{y}_{j} + 1$ \Comment{appending 1 to the right side of the codeword $\mathbf{y}_j$}
    \EndFor
    \State $\mathcal{B}_{2i-1}$ \texttt of $\mathcal{D}_{l,k+1} \gets [\mathbf{V};\mathbf{Z}]$ \Comment{concatenating all codewords in $\mathbf{V}$ and $\mathbf{Z}$}
    \State $\mathcal{B}_{2i}$ \texttt of $\mathcal{D}_{l,k+1} \gets \mathbf[\mathbf{W};\mathbf{U}]$ \Comment{concatenating all codewords in $\mathbf{W}$ and $\mathbf{U}$}
\EndFor
\end{algorithmic}
\end{algorithm}
In this coding algorithm, the coding system will keep the number of data bits ($k$) constant and increases the number of overhead bits ($l$), resulting in higher $n$. The increment of $l$ should be kept to a maximum of $1$ in each step, similar to the RASBA algorithm on the increment of $k$. Because of this, there will always be the same number of bins in each move, but every move will result in a doubling of the number of codewords in each bin. 

Codewords in two bins $\mathcal{B}_{2i-1}$ and $\mathcal{B}_{2i}$, of the current form $(l,k)$, will generate the same two bins in the $\mathcal{D}_{l+1,k}$. To accomplish this, $0$ and $1$ will be appended to the right side of all codewords in the bins $\mathcal{B}_{2i-1}$ and $\mathcal{B}_{2i}$ of the current form, respectively. The system can then produce $\mathcal{B}_{2i-1}$ of the next form $(l+1,k)$ by concatenating the new codewords produced by codewords in the bin $\mathcal{B}_{2i}$ to the end of the new codewords produced by $\mathcal{B}_{2i-1}$ of the current form. Similarly, bin $\mathcal{B}_{2i}$ of the next form $(l+1,k)$ will be created by concatenating new codewords, with the exception that $1$ will be appended to all codewords of bin $\mathcal{B}_{2i-1}$, and $0$ will be appended to all codewords of bin $\mathcal{B}_{2i}$ of the current form. This procedure will be repeated for all $i$, when $1 \leq i \leq \frac{2^k}{2}$.
The steps of the RAHBA algorithm are described in Algorithm~\ref{alg:rahba}.

As previously stated, the RAHBA algorithm allows the system to construct $\mathcal{D}_{l,k}$ from the $\mathcal{D}_{l-1,k}$, so the construction structure of $\mathcal{D}$ is as shown in Fig. \ref{RAHBA_Growth}, where $\mathcal{D}_{0,1}$ is the FF of the case when $l = 0$. 
\begin{figure}[ht]
\centering
\begin{tikzpicture}
\node [xshift=0.5cm, yshift=5cm] {$\mathcal{D}_{0,1}$};
\draw[-Stealth] (1,5) -- (1.5,5);
\node [xshift=2cm, yshift=5cm] {$\mathcal{D}_{1,1}$};
\draw[-Stealth] (2.5,5) -- (3,5);
\node [xshift=3.5cm, yshift=5cm] {$\mathcal{D}_{2,1}$};
\draw[-Stealth] (4,5) -- (4.5,5);
\node [xshift=5cm, yshift=5cm] {$\mathcal{D}_{3,1}$};
\draw[-Stealth] (5.5,5) -- (6,5);
\node [xshift=6.5cm, yshift=5cm] {$\mathcal{D}_{4,1}$};
\draw[-Stealth] (7,5) -- (7.5,5);
\node [xshift=8cm, yshift=5cm]
{$\dots$};

\node [xshift=0.5cm, yshift=4.5cm] {$\mathcal{D}_{1,2}$};
\draw[-Stealth] (1,4.5) -- (1.5,4.5);
\node [xshift=2cm, yshift=4.5cm] {$\mathcal{D}_{2,2}$};
\draw[-Stealth] (2.5,4.5) -- (3,4.5);
\node [xshift=3.5cm, yshift=4.5cm] {$\mathcal{D}_{3,2}$};
\draw[-Stealth] (4,4.5) -- (4.5,4.5);
\node [xshift=5cm, yshift=4.5cm] {$\mathcal{D}_{4,2}$};
\draw[-Stealth] (5.5,4.5) -- (6,4.5);
\node [xshift=6.5cm, yshift=4.5cm] {$\mathcal{D}_{5,2}$};
\draw[-Stealth] (7,4.5) -- (7.5,4.5);
\node [xshift=8cm, yshift=4.5cm]
{$\dots$};
\node [xshift=3.75cm, yshift=4cm]
{$\vdots$};
\end{tikzpicture}
\caption{Generating different forms by RAHBA method.}
  \label{RAHBA_Growth}
\end{figure}
It should be noted that, in Algorithms~\ref{alg:rasba} and \ref{alg:rahba}, capital letters (e.g. $V$) represent a $2^l \times 1$ column vectors, and small letters (e.g. $v$) denote each row of corresponding column vectors.
\subsection{Our Standard Path and Some Examples} \label{Our Standard Path and some examples}
The standard method used in this study to derive the $\mathcal{D}_{\tilde{l},\tilde{k}}$ from the zero point is based on starting from $l=0$ and $k=1$, then moving on the red line by RAHBA to reach the FF of the case $\tilde{l}$, and then moving to the direction of RASBA approach to achieve $\tilde{k}$, as illustrated in Fig~\ref{Schememap}. However, it is not the only path, and there are other options as well. In fact, by moving only on solid lines, the $\mathcal{D}$ of every intersection can be obtained from the $\mathcal{D}$ of every intersection whose $k$ and $l$ are less than or equal to the $l$ and $k$ of the final intersection, respectively. 
If $\mathcal{D}_{l_1,k_1}$ is the starting point and the goal is to reach $\mathcal{D}_{l_2,k_2}$ (for $\mathcal{D}$s with $l \geq 1$), then there are various ways to get there, where $n_1=k_1+l_1$ and $n_2=k_2+l_2$.
\begin{equation}
\binom{(l_2-l_1)+(k_2-k_1)}{(k_2-k_1)}=\binom{n_2-n_1}{k_2-k_1}=\binom{n_2-n_1}{l_2-l_1}.
\end{equation}
The $\mathcal{D}$ obtained via the standard path is referred to as the standard  $\mathcal{D}$. Below are a few illustrations of standard $\mathcal{D}$s.

\ifdefined\OneCol
 \begin{equation*}
\begin{aligned}[c]
&\mathcal{D}_{0,1}=
  \begin{bmatrix}
    0 & 1\\
  \end{bmatrix}  \\
&\mathcal{D}_{1,1}=
  \begin{bmatrix}
    00 & 01\\
    11 & 10\\
  \end{bmatrix}
\end{aligned}
\qquad
\begin{aligned}[c]
\mathcal{D}_{2,1}=
  \begin{bmatrix}
    000 & 001\\
    110 & 111\\
    011 & 010\\
    101 & 100\\
  \end{bmatrix}
\end{aligned}
\qquad
\begin{aligned}[c]
\mathcal{D}_{2,2}=
  \begin{bmatrix}
    0000 & 1110 & 0001 & 1111\\
    1011 & 1001 & 1010 & 1000\\
    1100 & 0010 & 1101 & 0011\\
    0111 & 0101 & 0110 & 0100\\
  \end{bmatrix}
\end{aligned}
\end{equation*}

\[
\mathcal{D}_{1,2}=
  \begin{bmatrix}
    000 & 001 & 010 & 011\\
    111 & 110 & 101 & 100\\
  \end{bmatrix}
\mathcal{D}_{1,3}=
  \begin{bmatrix}
  \scriptstyle  0000 & \scriptstyle 0001 & \scriptstyle 0010 & \scriptstyle 0011 & \scriptstyle 0100 & \scriptstyle 0101 & \scriptstyle 0110 & \scriptstyle 0111\\
   \scriptstyle 1111 & \scriptstyle 1110 & \scriptstyle 1101 & \scriptstyle 1100 & \scriptstyle 1011 & \scriptstyle 1010 & \scriptstyle 1001 & \scriptstyle 1000\\
  \end{bmatrix}
\]

\newcommand\scalemath[2]{\scalebox{#1}{\mbox{\ensuremath{\displaystyle #2}}}}

\[
\scalemath{0.9}{
\mathcal{D}_{3,1}=
  \begin{bmatrix}
    0000 & 0001\\
    0011 & 0010\\
    0110 & 0100\\
    1100 & 1000\\
    1001 & 0111\\
    1010 & 1011\\
    0101 & 1101\\
    1111 & 1110\\
  \end{bmatrix} 
\mathcal{D}_{3,2}=
  \begin{bmatrix}
    00000 & 00010 & 00001 & 00011\\
    00111 & 00101 & 00110 & 00100\\
    01100 & 01000 & 01101 & 01001\\
    11001 & 10001 & 11000 & 10000\\
    10010 & 01110 & 10011 & 01111\\
    10101 & 10111 & 10100 & 10110\\
    01010 & 11010 & 01011 & 11011\\
    11111 & 11101 & 11110 & 11100\\
  \end{bmatrix}
  }
\]

\begin{equation*}
\mathcal{D}_{4,1}=
    \rotatebox{90}{$\begin{bmatrix}
    00000 & 00001\\
    00110 & 00111\\
    01100 & 01101\\
    11000 & 11001\\
    10010 & 10011\\
    10100 & 10101\\
    01010 & 01011\\
    11110 & 11111\\
    00011 & 00010\\
    00101 & 00100\\
    01001 & 01000\\
    10001 & 10000\\
    01111 & 01110\\
    10111 & 10110\\
    11011 & 11010\\
    11101 & 11100\\
\end{bmatrix}$}
\end{equation*}
\begin{equation} \label{DsofNAs}
\end{equation}
\fi
\ifdefined\TwoCol
 \begin{equation*}
\begin{aligned}[c]
&\mathcal{D}_{0,1}=
  \begin{bmatrix}
    0 & 1\\
  \end{bmatrix}  \\
&\mathcal{D}_{1,1}=
  \begin{bmatrix}
    00 & 01\\
    11 & 10\\
  \end{bmatrix}
\end{aligned}
\qquad
\begin{aligned}[c]
\mathcal{D}_{2,1}=
  \begin{bmatrix}
    000 & 001\\
    110 & 111\\
    011 & 010\\
    101 & 100\\
  \end{bmatrix}
\end{aligned}
\end{equation*}
\begin{equation*}
\begin{aligned}[c]
\mathcal{D}_{2,2}=
  \begin{bmatrix}
    0000 & 1110 & 0001 & 1111\\
    1011 & 1001 & 1010 & 1000\\
    1100 & 0010 & 1101 & 0011\\
    0111 & 0101 & 0110 & 0100\\
  \end{bmatrix}
\end{aligned}
\end{equation*}
\[
\mathcal{D}_{1,2}=
  \begin{bmatrix}
    000 & 001 & 010 & 011\\
    111 & 110 & 101 & 100\\
  \end{bmatrix}
  \]
  \[
\mathcal{D}_{1,3}=
  \begin{bmatrix}
  \scriptstyle  0000 & \scriptstyle 0001 & \scriptstyle 0010 & \scriptstyle 0011 & \scriptstyle 0100 & \scriptstyle 0101 & \scriptstyle 0110 & \scriptstyle 0111\\
   \scriptstyle 1111 & \scriptstyle 1110 & \scriptstyle 1101 & \scriptstyle 1100 & \scriptstyle 1011 & \scriptstyle 1010 & \scriptstyle 1001 & \scriptstyle 1000\\
  \end{bmatrix}
\]

\newcommand\scalemath[2]{\scalebox{#1}{\mbox{\ensuremath{\displaystyle #2}}}}

\[
\scalemath{0.9}{
\mathcal{D}_{3,1}=
  \begin{bmatrix}
    0000 & 0001\\
    0011 & 0010\\
    0110 & 0100\\
    1100 & 1000\\
    1001 & 0111\\
    1010 & 1011\\
    0101 & 1101\\
    1111 & 1110\\
  \end{bmatrix} 
\mathcal{D}_{3,2}=
  \begin{bmatrix}
    00000 & 00010 & 00001 & 00011\\
    00111 & 00101 & 00110 & 00100\\
    01100 & 01000 & 01101 & 01001\\
    11001 & 10001 & 11000 & 10000\\
    10010 & 01110 & 10011 & 01111\\
    10101 & 10111 & 10100 & 10110\\
    01010 & 11010 & 01011 & 11011\\
    11111 & 11101 & 11110 & 11100\\
  \end{bmatrix}
  }
\]

\begin{equation*}
\mathcal{D}_{4,1}=
    \rotatebox{90}{$\begin{bmatrix}
    00000 & 00001\\
    00110 & 00111\\
    01100 & 01101\\
    11000 & 11001\\
    10010 & 10011\\
    10100 & 10101\\
    01010 & 01011\\
    11110 & 11111\\
    00011 & 00010\\
    00101 & 00100\\
    01001 & 01000\\
    10001 & 10000\\
    01111 & 01110\\
    10111 & 10110\\
    11011 & 11010\\
    11101 & 11100\\
\end{bmatrix}$}
\end{equation*}
\begin{equation} \label{DsofNAs}
\end{equation}
\fi
As previously mentioned, it is possible that different paths yield different $\mathcal{D}$ for a specific form, but all of them belong to the same family of code distribution (CD). As far as we observed, these different $\mathcal{D}$s may produce the same equivocation or a very slight, barely noticeable variation.

\section{The solutions of optimization for various cases} \label{The Answer of optimization for different cases}
In this section, the results of (\ref{optimization}) are examined for various cases with a different number of overhead bits $(l)$. We also explain how the solutions of (\ref{optimization}) helped us to develop the suggested Ni coding discussed in the previous section. 

\subsection{One overhead bit case ($\mathit{l=1}$)} \label{l1}
In this case, the coding system uses just one overhead bit. As an example, assume that $n=5$ and $k=4$. According to this assumption the potential rows of $\mathbf{R}$, discussed in section \ref{Equivocation calculation}, are  $\mathbf{r}_i \in (\mathbb{Z}^+)^{6}$. In this example, the coding system should choose $16$ rows from all potential rows. Note that a specific row can be selected more than once. The optimum answer, shown by $\mathbf{x}^*$, indicates that the maximum value of (\ref{Conditioned_Ent}) can be obtained by choosing one $[1,0,0,0,0,1]$ row, five $[0,1,0,0,1,0]$ rows, and ten $[0,0,1,1,0,0]$ rows from all potential rows. For example, $[1,0,0,0,0,1]$ indicates that there are two codewords in the bin, one with Hamming distance $0$ with the $z^n$ and one with Hamming distance five with $z^n$. By evaluating this pattern, we find that this is the optimum answer for the form $(1,4)$ and all $p$. The following is the approved pattern of choosing rows for other forms, i.e., other values of $k$. There are $\binom{n}{0}$ rows of $[1,0,\hdots,0,1]$, $\binom{n}{1}$ rows of $[0,1,0,\hdots,0,1,0]$, to finally $\binom{n}{ \DeclarePairedDelimiter\floor{\lfloor}{\rfloor} \floor{\frac{n}{2}}}$ rows of $[0,\hdots,0,1,1,0,\hdots,0]$ for the odd value of $n$, and $\binom{n}{\frac{n}{2}}$ rows of $[0,\hdots,0,2,0,\hdots,0]$ for even value of $n$. 
As a result, in the general instances when $l=1$, this pattern works as the optimal solution and solves the problem for all values $n$, $k$, and $p$.

As was covered in the final part of Section~\ref{Optimization}, there is a new limit that can be determined by solving the LP problem for various values of $p$. The LP-derived limit is depicted in Fig. \ref{nminus1curves} for the form (1,4) and compared with other results in Section \ref{Simulation results}.

\begin{lemma}
By using our suggested Ni coding, the proposed limit is reachable for all the forms when $l=1$.
\end{lemma}

\begin{proof}
To prove it, we need to show that there is a code that results in the $2^n$ $\mathbf{R}$ of the optimum answer when each of its codewords XORs with each of $2^n$ $z^n$s. Then, the total equivocation, which is the average of these $2^n$ values, will produce the maximum and optimum result.

When a $z^n$ is added to the codewords of the code table, the bits of the codeword corresponding to $1$ in the $z^n$, will flip. The optimum solution in this case, is the one with an $i$ Hamming distance and an $n-i$ Hamming distance with $z^n$. Flipping each subset of bits of all the codewords in the code table is equivalent to their Hamming distances to each $z^n$. There is only one distribution for the code table such that after flipping every certain subset of bits of all $2^n$ codewords, all $2^n$ $\mathbf{R}$ matrices would be the solution for the case when $l=1$. This occurs as a result of the placement of each codeword and its unique \textit{opposite codeword} in a bin. The codeword that has all of its bits reversed is said to be the opposite codeword. In each bin, every pair of corresponding bits of the two codewords (e.g. the $i$-th bit in each) are opposite. Therefore, if none or some or all of the bits of two opposite codewords are flipped, the two flipped codewords (considering them as Hamming distances with a $z^n$) will also be opposite, and thus, will have the pattern and the matrix $\mathbf{R}$ of the optimum answer. Basically, these code tables of opposite codewords for $l=1$, are the proposed Ni codes for the case when $l=1$. The three code tables of $\mathcal{D}_{1,1}$, $\mathcal{D}_{1,2}$, and $\mathcal{D}_{1,3}$ have been shown in Section \ref{Our Standard Path and some examples}.
\end{proof}

This concept of opposite codewords was presented to support the analytical proof and shows how the suggested Ni coding scheme is able to attain the desired LP-derived limits for $l=1$. The cohesive Ni coding scheme is the great picture that this opposite notion fits in it when $l=1$. The $\mathcal{D}$s of this case, can be obtained by the RASBA technique and moving on the horizontal line where $l=1$ to increase $k$ as illustrated in Fig. \ref{Schememap}.

\subsection{More than one overhead bit cases ($\mathit{l=2,3,4}$, and etc.)} \label{l234}
For the other cases when the coding system uses more than one overhead bit, e.g., $l=2,3,4$, different  $\mathbf{x^*}$ for various values of $p$, can be obtained by solving (\ref{optimization}) for these specific cases. The solutions produce curves, each representing each form's LP-derived limit. Four of these LP-derived limits, each limit for a form as well as each form for a different case, are displayed and compared with other results in Section \ref{Simulation results}.

In these cases, the code tables generated by our suggested Ni coding technique, produce the best secrecy code with the highest equivocation, among numerous random bin codes.  
In the FF of the case $l=2$, which is $(2,1)$, the challenge is the distribution of all eight codewords into two bins (groups) of four codewords. There are a total of $70$ distinct ways $\binom{8}{4}=70$ to accomplish this distribution. One of these distributions, which later proved to be the Ni code for the form $(2,1)$, appeared to produce the maximum total equivocation (\ref{equivovation_t}), for all values of $p$, as it is shown in Section \ref{Our Standard Path and some examples}. This combination of distributions has two characteristics: first, in every cell, the summation of all the binary codewords is a $\mathbf{0}$ vector. The second feature is that the complementary of each codeword in a bin, in terms of the number of  $1$s, exists in the other bin. It can be argued that these characteristics help the Ni coding scheme work well in this situation. 
To find and extend these two characteristics for higher values of $n$ in the current case, the RASBA scheme was developed. We observed that these two characteristics by using Ni's code table of $(2,2)$ will be maintained, which is displayed as $\mathcal{D}_{2,2}$; and from that to $\mathcal{D}$s of higher forms of this case.

When $l=3$ and $l=4$, the distribution procedure is the same as when $l=2$. As a result, the best binning for their FF are those shown as $\mathcal{D}_{3,1}$ and $\mathcal{D}_{4,1}$, respectively, in Section \ref{Our Standard Path and some examples}. Once more, using the RASBA technique on small forms results in higher forms of Ni coding on their specific cases. For example, it is feasible to reach $\mathcal{D}_{3,2}$  from $\mathcal{D}_{3,1}$, or $\mathcal{D}_{3,3}$ from $\mathcal{D}_{3,2}$. Continued use of this strategy will result in higher $\mathcal{D}$s of these cases. Furthermore, the RAHBA scheme was developed to generate the $\mathcal{D}$ of the FF of a case from the $\mathcal{D}$ of the FF of its lower case, just by moving on the red line as illustrated in Fig. \ref{Schememap}. This was done by analyzing the connections between the $\mathcal{D}$s of the FF on various cases. 

Apart from all these, for those cases and forms of Ni codes that have linear properties, the matrices $\mathbf{G}$ and $\mathbf{H}^T$ for encoding and decoding procedures have been explored in Section \ref{Matrices to generate the code}.
\subsection{Simulation Results} \label{Simulation results}
In a nutshell, the secrecy capacity of the system model, described in Section \ref{System Model} is as follows
\begin{equation}  \label{wiretapcapacityvalue}
C_s = 1 - (1-h(P) )=h(P).
\end{equation}
Wyner in \cite{Wyner1975}, showed that
\begin{equation}  \label{wynerlimit} 
\frac{H(M|Z^n)}{n}  \leq C_s.
\end{equation}
Since conditioning always reduces entropy, another limit for normalized equivocation is
\begin{equation}  \label{coverlimit}
0 \leq \frac{H(M|Z^n)}{n} \leq \frac{H(M)}{n}=\frac{k}{n}=R,
\end{equation} 
and from (\ref{wiretapcapacityvalue}), (\ref{wynerlimit}), and (\ref{coverlimit}) we have
\begin{equation}    \label{bothlimit}
0 \leq \frac{H(M|Z^n)}{n} \leq \min  \{h(p),R\}.
\end{equation}
In the rest of the paper, we refer to the bound of (\ref{bothlimit}) as the \textit{Infinite Blocklength Limit}.

To evaluate the performance of our proposed Ni code, the equivocation of our coding scheme is compared with other random binning codes. Moreover, we compared the aforementioned equivocation rates and the LP-derived limits with the infinite blocklength limit (\ref{bothlimit}).

As shown in the following figures, as $n$ increases, the interval of $p$, in which the $h(p)$ is determinative in (\ref{bothlimit}), expands, and thus the value of $p$ at which the two bounds in (\ref{bothlimit}) intersect tends to $0.5$.

\ifdefined\OneCol
\begin{figure}[ht]
\centering
\begin{minipage}[b]{0.48\linewidth}
\centerline{\includegraphics[width=\textwidth]{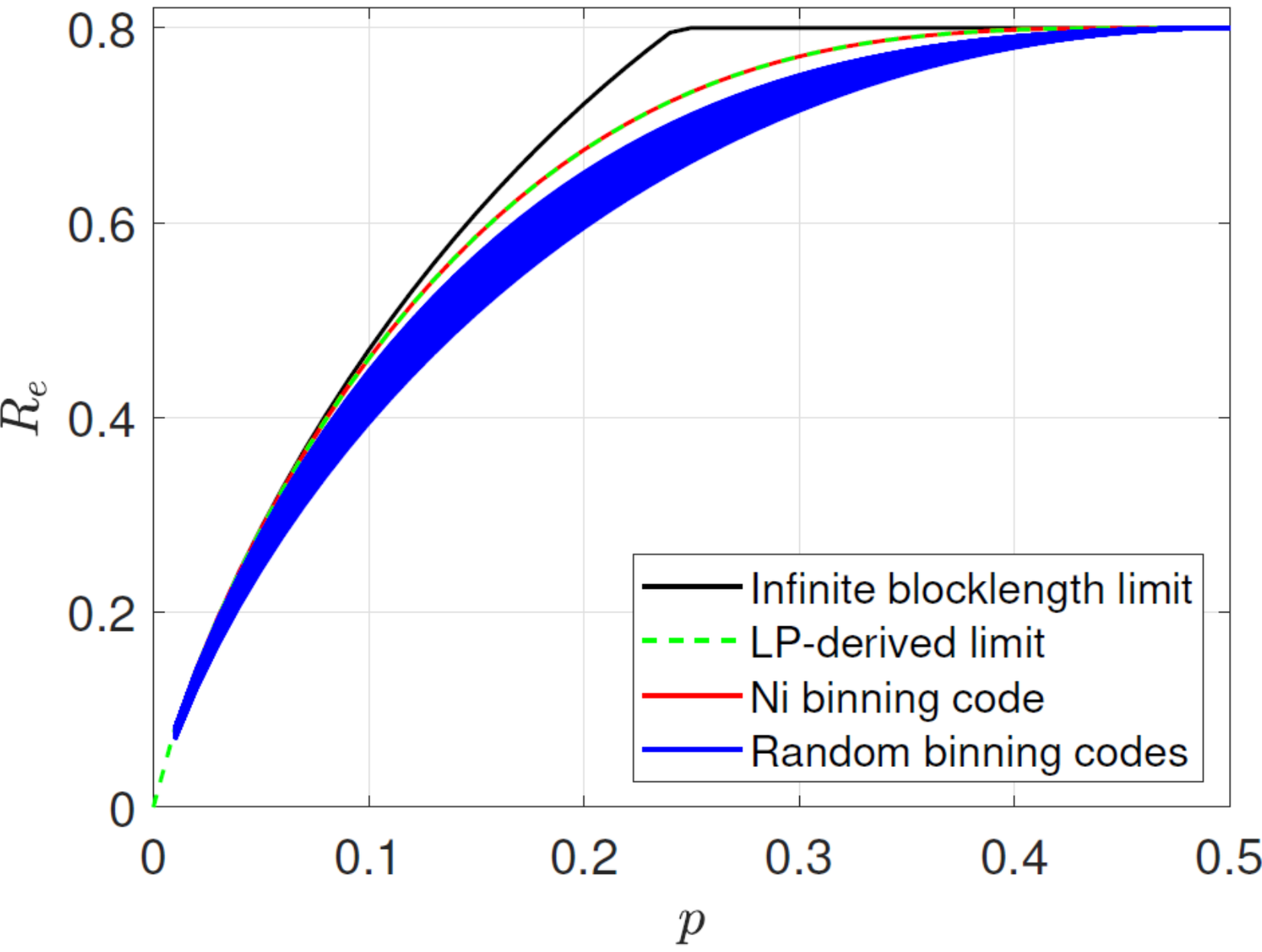}}
\caption{Infinite blocklength limit, LP-derived limit, and equivocation results of the Ni code and different binning codes for $l=1$ and $n=5$.}
\label{nminus1curves}
\end{minipage}
\quad
\begin{minipage}[b]{0.48\linewidth}
\centerline{\includegraphics[width=\textwidth]{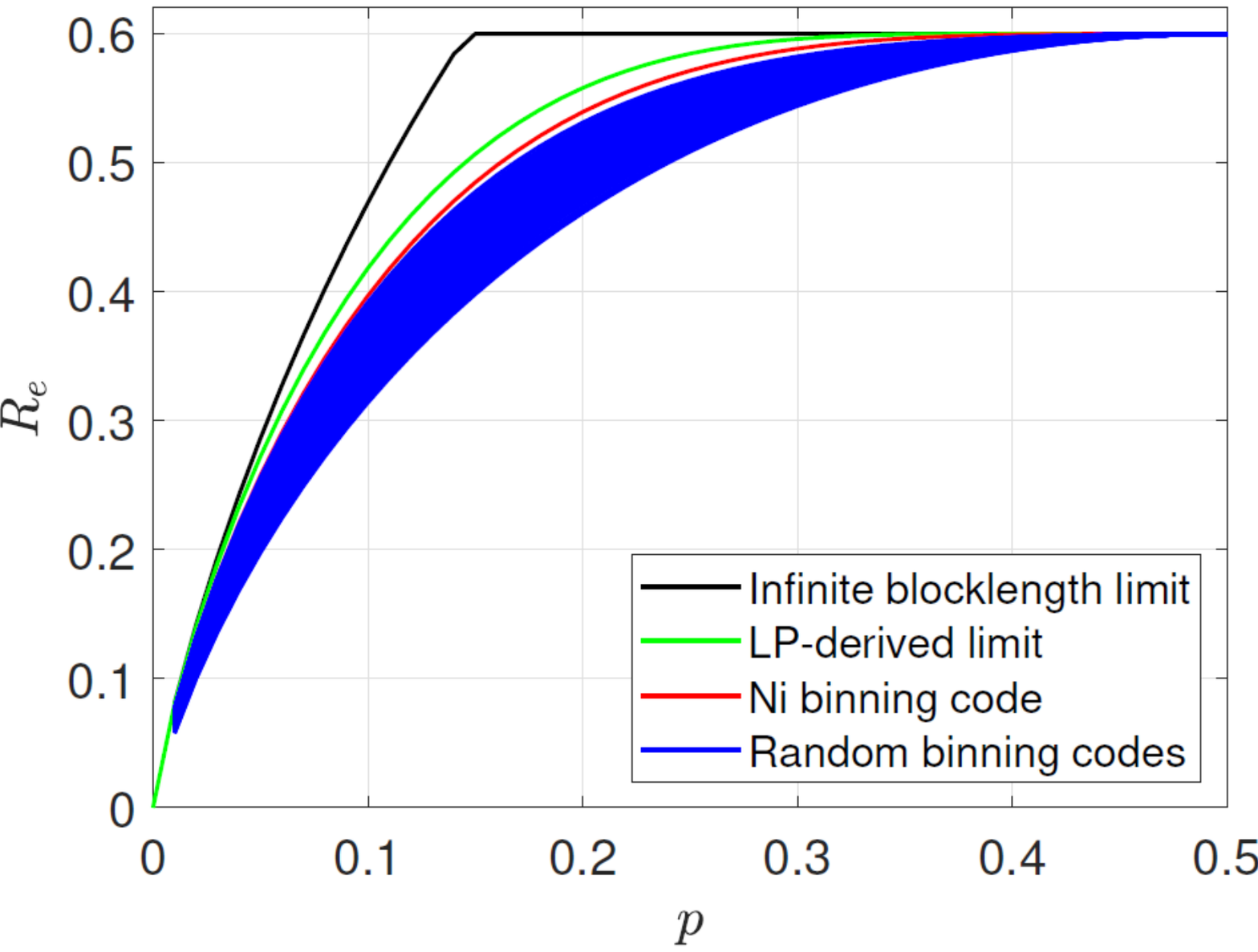}}
\caption{Infinite blocklength limit, LP-derived limit, and equivocation results of the Ni code and different binning codes for $l=2$ and $n=5$.}
\label{fignminus2}
\end{minipage}
\end{figure}

\begin{figure}[ht]
\centering
\begin{minipage}[b]{0.48\linewidth}
\centerline{\includegraphics[width=\textwidth]{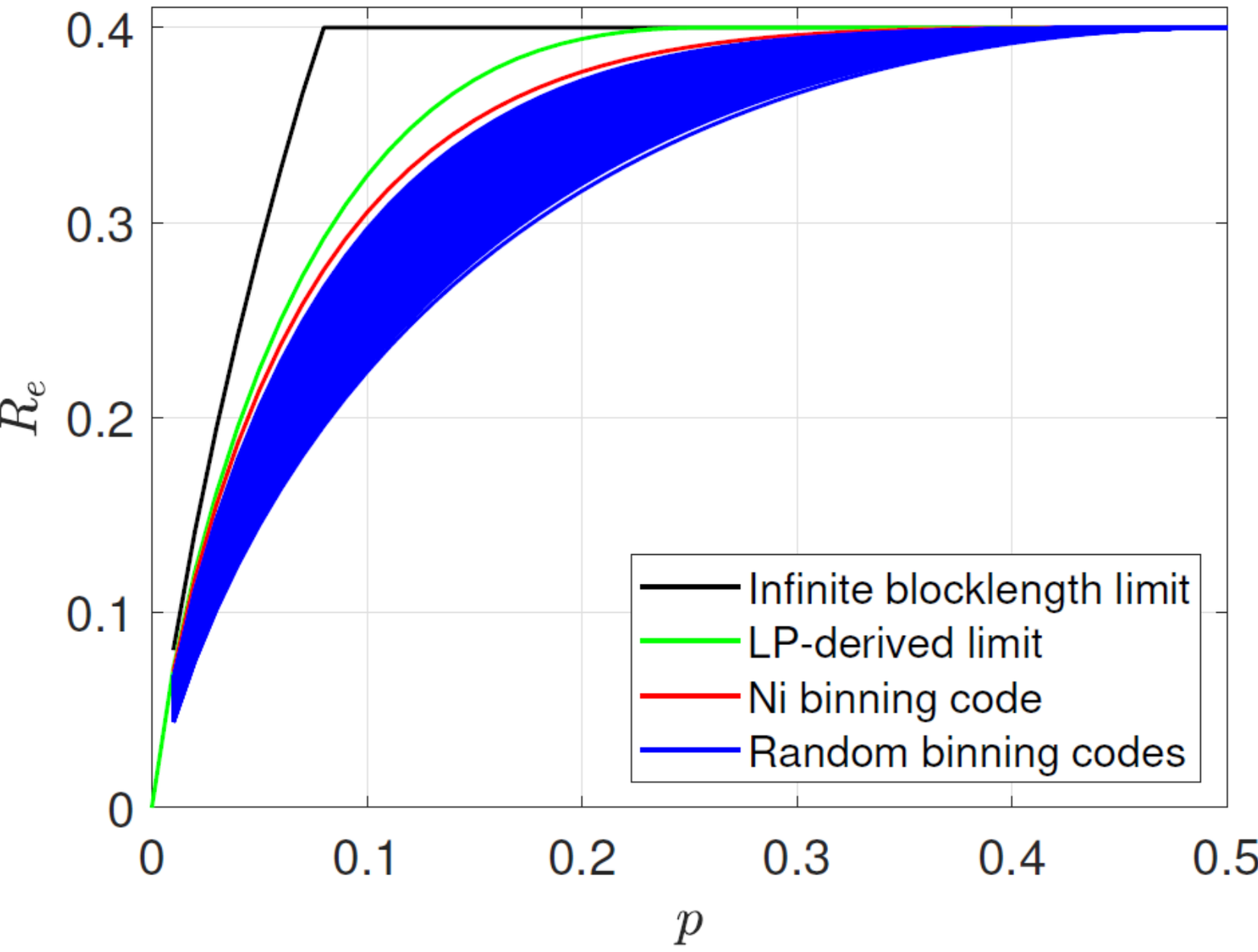}}
\caption{Infinite blocklength limit, LP-derived limit, and equivocation results of the Ni code different binning codes for $l=3$ and $n=5$.}
\label{nminus3curves}
\end{minipage}
\quad
\begin{minipage}[b]{0.48\linewidth}
\centerline{\includegraphics[width=\textwidth]{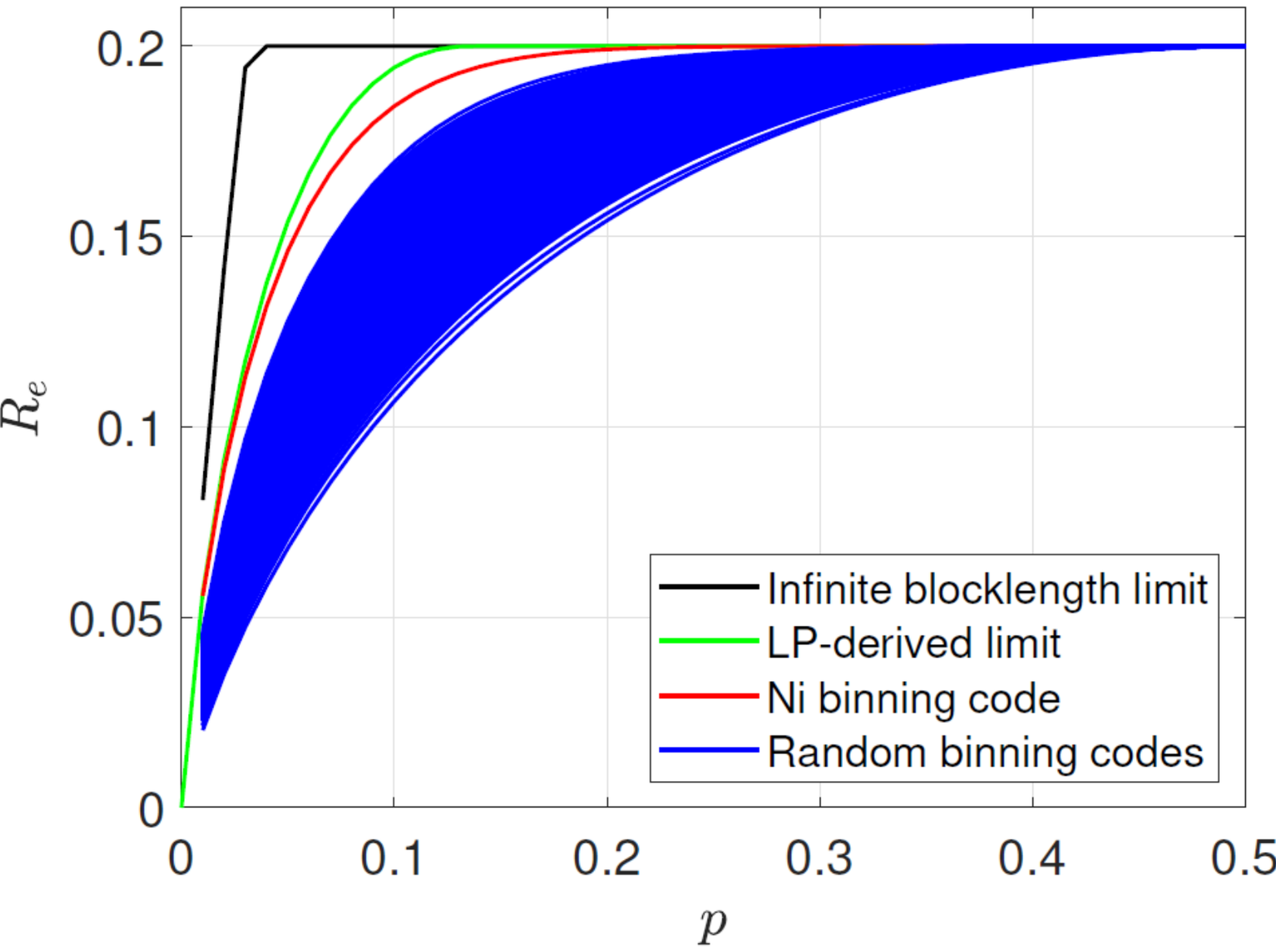}}
\caption{Infinite blocklength limit, LP-derived limit, and equivocation results of the Ni code and different binning codes for $l=4$ and $n=5$.}
\label{nminus4curves}
\end{minipage}
\end{figure}
\fi

\ifdefined\TwoCol
  \begin{figure}[htbp]
\centerline{\includegraphics[scale=0.35]{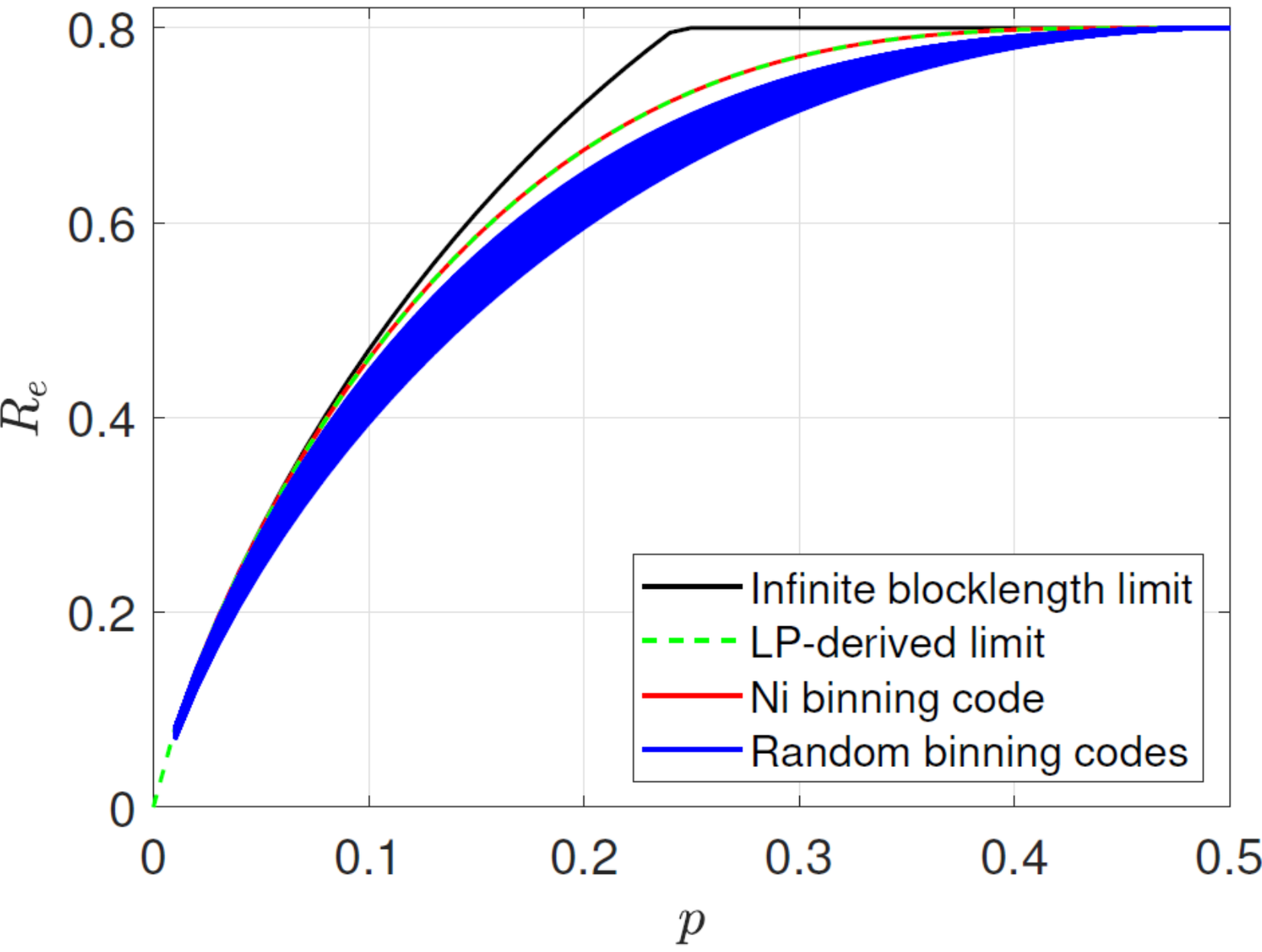}}
\caption{Infinite blocklength limit, LP-derived limit, and equivocation results of the Ni code and different binning codes for $l=1$ and $n=5$.}
\label{nminus1curves}
\end{figure}

\begin{figure}[tbp]
\centerline{\includegraphics[scale=0.35]{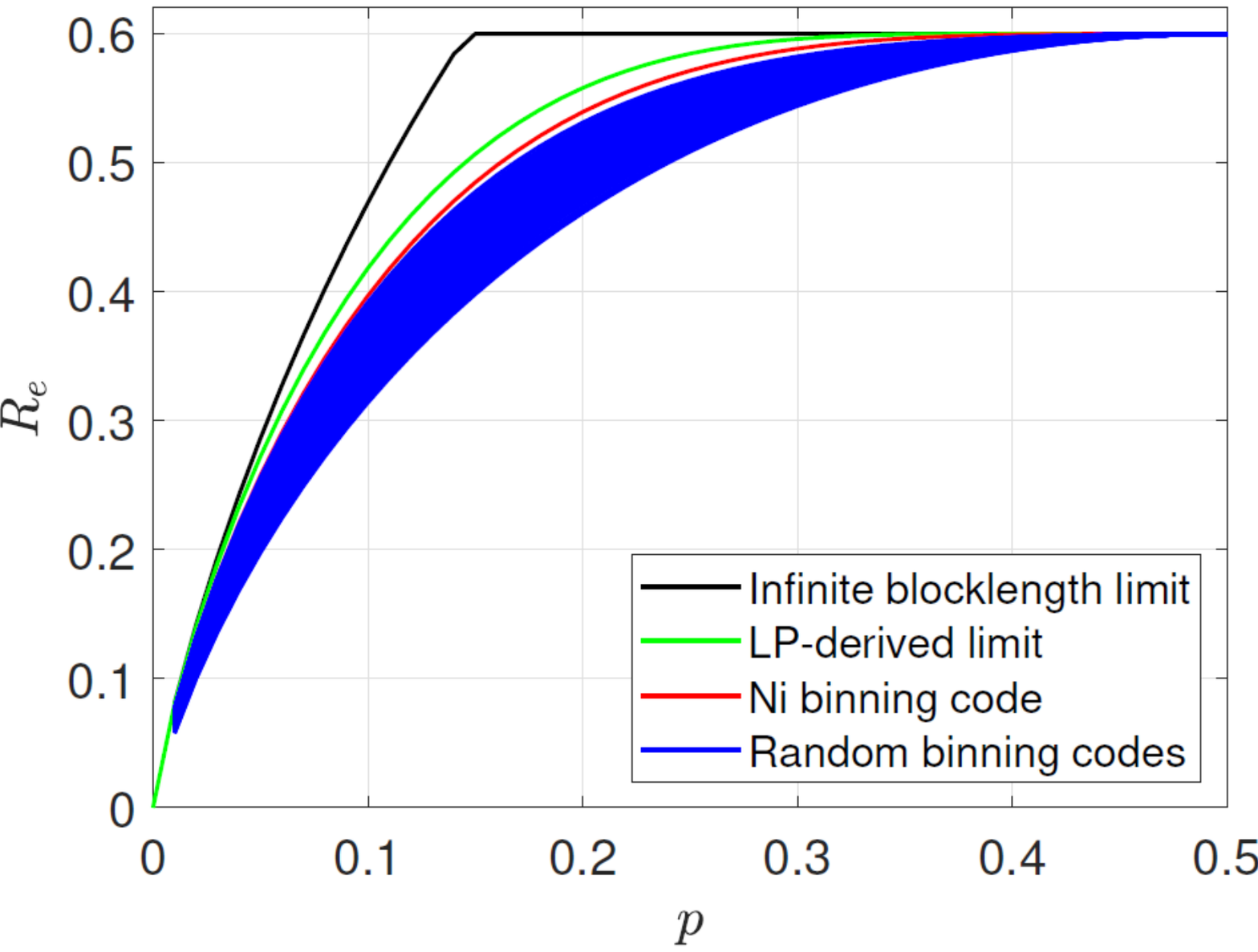}}
\caption{Infinite blocklength limit, LP-derived limit, and equivocation results of the Ni code and different binning codes for $l=2$ and $n=5$.}
\label{fignminus2}
\end{figure}

\begin{figure}[htbp]
\centerline{\includegraphics[scale=0.35]{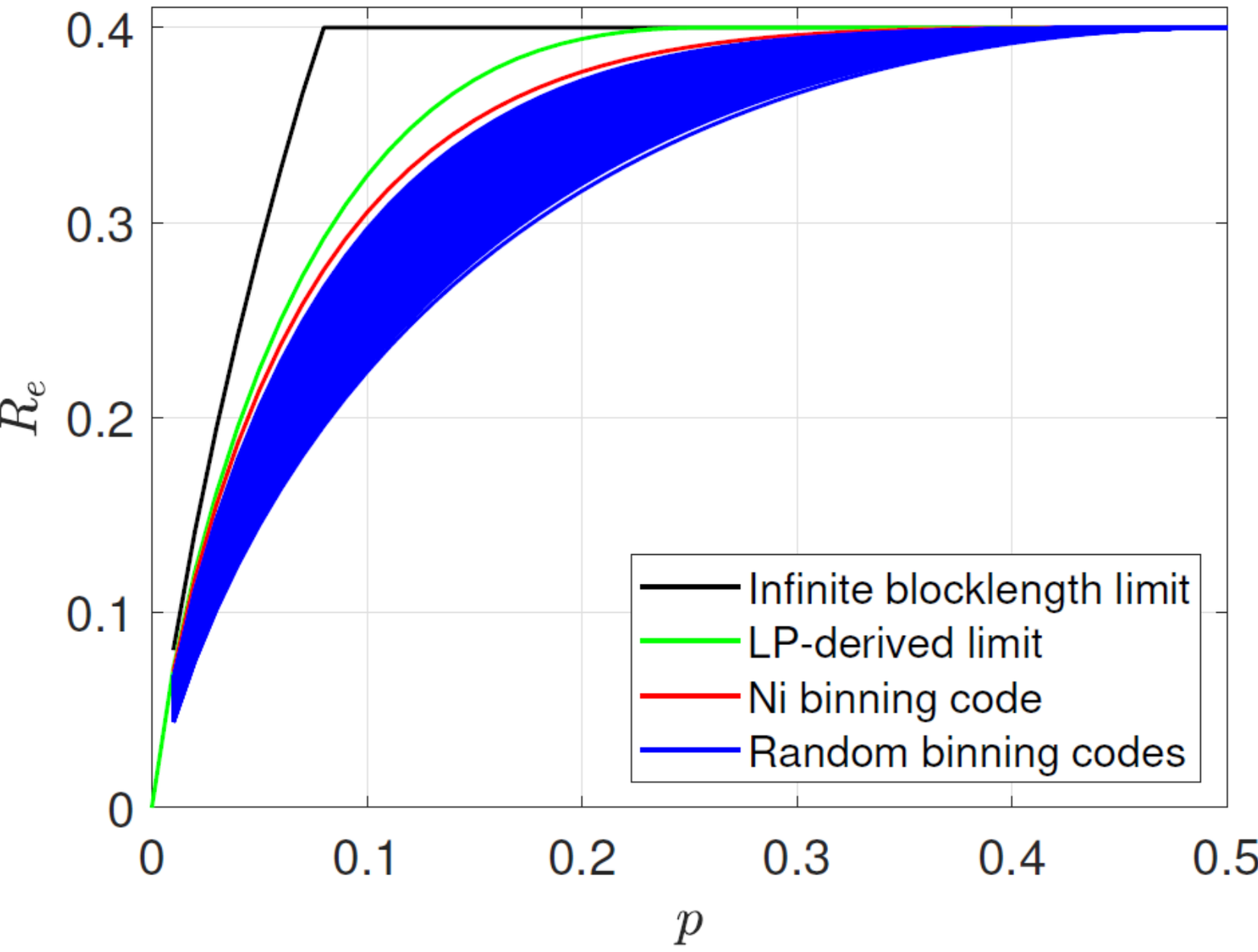}}
\caption{Infinite blocklength limit, LP-derived limit, and equivocation results of the Ni code and different binning codes for $l=3$ and $n=5$.}
\label{nminus3curves}
\end{figure}

\begin{figure}[htbp]
\centerline{\includegraphics[scale=0.35]{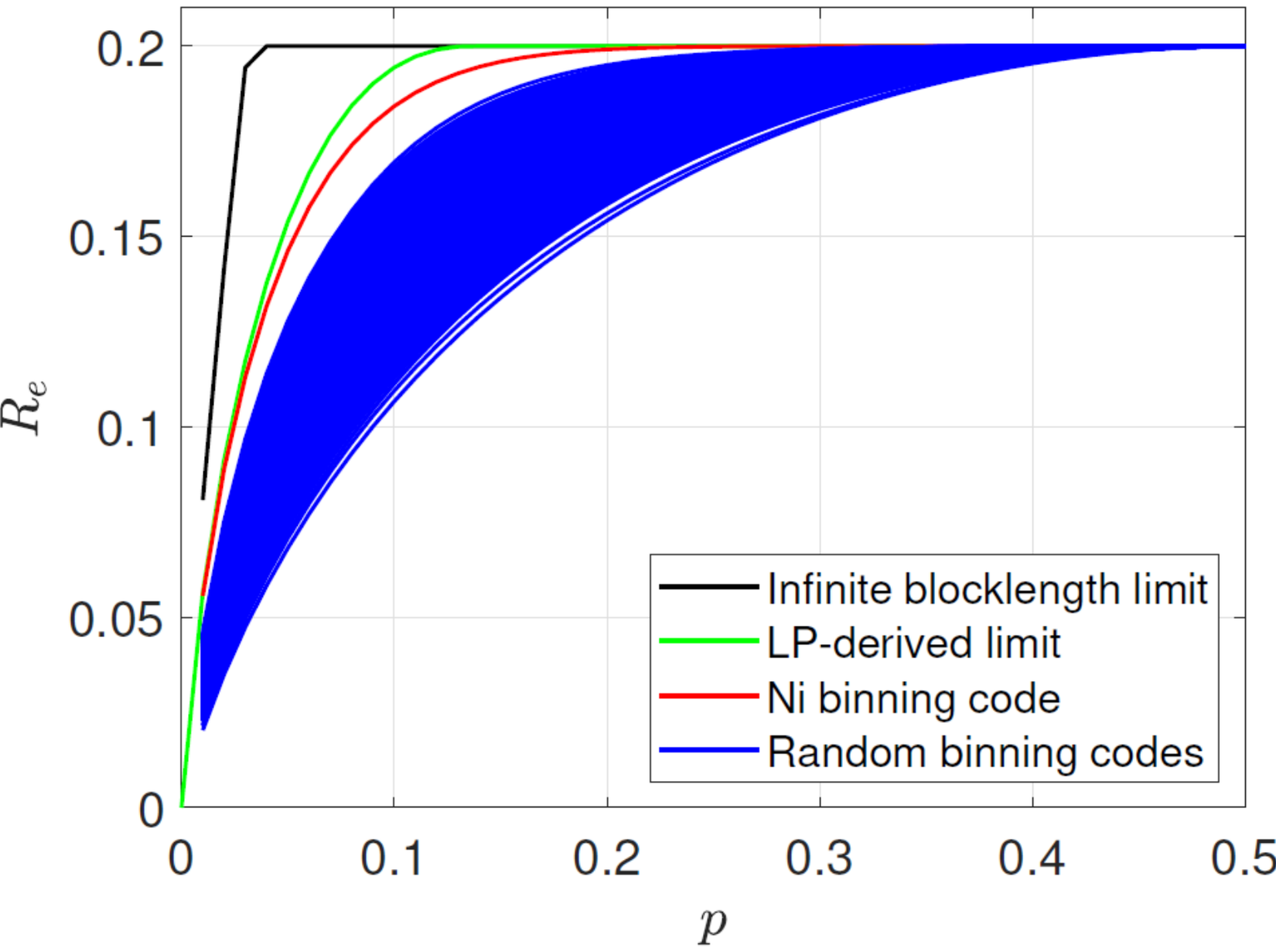}}
\caption{Infinite blocklength limit, LP-derived limit, and equivocation results of the Ni code and different binning codes for $l=4$ and $n=5$.}
\label{nminus4curves}
\end{figure}
\fi

We showed the performance of our coding approach and the LP-derived limits for four forms of $(1,4)$, $(2,3)$, $(3,2)$, and $(4,1)$, in figures \ref{nminus1curves}, \ref{fignminus2}, \ref{nminus3curves}, and \ref{nminus4curves},  respectively. In each of the figures, the black curve represents the infinite blocklength limit (\ref{bothlimit}), and the green curve indicates our presented LP-derived limit that can be obtained by solving (\ref{optimization}) for every $p$. Additionally, the red curve is the result of our Ni coding scheme from the standard $\mathcal{D}$. It can be seen in Fig. \ref{nminus1curves} that the green curve and the red curve coincide with each other since it is possible to achieve the LP-derived limit by the Ni coding scheme. Clearly, there are significant gaps between the red curves and the blue curves, which represent the equivocation rates determined by the Ni code and numerous random codes, respectively. Due to the symmetry of the results, the curves are shown for $0 \leq p \leq 0.5$.

We also derived the total number of random binning codes based on the binning structure, distributing all $2^n$ codewords uniformly in the bins, as follows
\begin{equation} \label{hugenumber}
\prod _{i=1}^{2^k} \binom{ei-1}{e-1}.
\end{equation}
Hence, for these four forms of $(1,4)$, $(2,3)$, $(3,2)$, and $(4,1)$, the total number of codes are $\prod _{i=1}^{16} \binom{2i-1}{1}\approx {1.92} \times 10^{17}$, $\prod _{i=1}^{8} \binom{4i-1}{3}\approx {5.93} \times 10^{19}$, $\prod _{i=1}^{4} \binom{8i-1}{7}\approx {4.15} \times 10^{15}$, and $\prod _{i=1}^{2} \binom{16i-1}{15}\approx {3.01} \times 10^{8}$, respectively. So, systematically simulating all of them is not really feasible. As a result, to provide accuracy in comparison, the number of generated random binning codes in each of the four examples is $10^5$ codes. Although, there was not a noticeable difference between the realization of $10^4$ and $10^5$, in the output of each of these examples. 
\section{Code Generation for linear cases} \label{Matrices to generate the code}
This section explains the construction of matrices $\mathbf{G}$ and $\mathbf{H}$, in order to encode and decode messages in the Ni coding scheme for linear cases. The $k$ bit message of $m^k$ is encoded by $x^n=[m^k \ v^l]\mathbf{G}$, where $v$ is a sequence with length $l$ that generates redundancy. The codeword is decoded as $\hat{m}^k= x^n \mathbf{H}^T$. The generator and the parity-check matrices are assumed to be known to Eve. Here, they are specified to a certain $(l,k)$ form with left and right super-scripts, i.e. ${}^l\mathbf{G}^k$ and ${}^l\mathbf{H}^k$.
The pattern to generate $\mathbf{G}$ and $\mathbf{H}$ for all forms of all even cases has been discovered. They are displayed with the left super-script of $ev$. For odd cases, the binning of the Ni code is not generally linear, and the matrices $\mathbf{G}$ and $\mathbf{H}$ are not achievable for these cases. However, there are two exceptions to this rule. First, the pattern of matrices $\mathbf{G}$ and $\mathbf{H}$ for all forms of the case $l=1$ are achievable, and it is interesting that the code for this case is optimum regarding equivocation. The matrices ${}^1\mathbf{G}$ and ${}^1\mathbf{H}$ are used to represent generator and parity-check matrices for the case $l=1$, respectively. The second exception is when $l = 3$ and $k = 1$, form $(3, 1)$, and generator and parity-check matrices of this form are also achievable. Furthermore, we observe that the following patterns, which are presented in both $\mathbf{G}$ and $\mathbf{H}^T$ can be extended to produce generator and parity-check matrices of other values of $n$ and $k$. To better understand these patterns, we illustrated generator and parity-check matrices of some examples at the end of this section.

The elements of matrices $\mathbf{G}$ and $\mathbf{H}^T$ at row $i$ and column $j$ can be shown as follows
\ifdefined\OneCol
 \begin{equation} \label{MatricesPatterns}
\begin{aligned}[c]
&[{}^{1}\mathbf{G}]_{i,j} = \begin{cases}
             0  & i=j \leq k \\
             1  & \text{others}
\end{cases}  \\
&[{}^{ev}\mathbf{G}]_{i,j} = \begin{cases}
             0  & 2 \leq i=j \\
             1  & \text{others}
       \end{cases} \\
&[{{}^{1}\mathbf{H}}^T]_{i,j} = \begin{cases}
             1  & i=j \leq k \\
             1  & i=k+1 , \forall j \\
             0  & \text{others}
\end{cases}
\end{aligned}
\qquad
\begin{aligned}[c]
[{{}^{ev}\mathbf{H}}^T]_{i,j} = \begin{cases}
             1  & i=1, \forall j \geq 2  \\
             1  & j=1, \forall i \geq 2 \\
             1  & 2 \leq i=j \leq k \\
             1  & i=j=1\\  &\text{ for odd} \ k\text{s} \\
             0  & i=j=1\\  &\text{ for even} \   k\text{s} \\
             0  & \text{others}
       \end{cases}
\end{aligned}
\end{equation}
\fi
\ifdefined\TwoCol
\begin{equation} \label{MatricesPatterns}
\begin{aligned}[c]
&[{}^{1}\mathbf{G}]_{i,j} = \begin{cases}
             0  & i=j \leq k \\
             1  & \text{others}
\end{cases}  \\
&[{}^{ev}\mathbf{G}]_{i,j} = \begin{cases}
             0  & 2 \leq i=j \\
             1  & \text{others}
       \end{cases} \\
&[{{}^{1}\mathbf{H}}^T]_{i,j} = \begin{cases}
             1  & i=j \leq k \\
             1  & i=k+1 , \forall j \\
             0  & \text{others}
\end{cases}
\end{aligned}
\hspace{-13pt}
\begin{aligned}[c]
[{{}^{ev}\mathbf{H}}^T]_{i,j} = \begin{cases}
             1  & i=1, \forall j \geq 2  \\
             1  & j=1, \forall i \geq 2 \\
             1  & 2 \leq i=j \leq k \\
             1  & i=j=1\\  &\text{ for odd} \ k\text{s} \\
             0  & i=j=1\\  &\text{ for even} \   k\text{s} \\
             0  & \text{others}
       \end{cases}
\end{aligned}
\end{equation}
\fi
The following examples of $\mathbf{G}_{n \times n}$ and $\mathbf{H}^T_{n \times k}$ illustrate these construction patterns
\ifdefined\OneCol
\[
{}^1\mathbf{G}^4=
  \begin{bmatrix}
    0 & 1 & 1 & 1 & 1\\
    1 & 0 & 1 & 1 & 1\\
    1 & 1 & 0 & 1 & 1\\
    1 & 1 & 1 & 0 & 1\\
    1 & 1 & 1 & 1 & 1\\
  \end{bmatrix}  \hspace{0.4em}
{{}^1\mathbf{H}^4}^T=
  \begin{bmatrix}
   1 & 0 & 0 & 0\\
   0 & 1 & 0 & 0\\
   0 & 0 & 1 & 0\\
   0 & 0 & 0 & 1\\
   1 & 1 & 1 & 1\\
  \end{bmatrix}  \hspace{0.4em}
{}^2\mathbf{G}^3=
  \begin{bmatrix}
    1 & 1 & 1 & 1 & 1\\
    1 & 0 & 1 & 1 & 1\\
    1 & 1 & 0 & 1 & 1\\
    1 & 1 & 1 & 0 & 1\\
    1 & 1 & 1 & 1 & 0\\
  \end{bmatrix}  \hspace{0.4em}
{{}^2\mathbf{H}^3}^T=
  \begin{bmatrix}
   1 & 1 & 1\\
   1 & 1 & 0\\
   1 & 0 & 1\\
   1 & 0 & 0\\
   1 & 0 & 0\\
  \end{bmatrix}  
\]
\[
{}^{4}\mathbf{G}^1=
  \begin{bmatrix}
    1 & 1 & 1 & 1 & 1\\
    1 & 0 & 1 & 1 & 1\\
    1 & 1 & 0 & 1 & 1\\
    1 & 1 & 1 & 0 & 1\\
    1 & 1 & 1 & 1 & 0\\
  \end{bmatrix}  \quad \quad
{{}^{4}\mathbf{H}^1}^T=
  \begin{bmatrix}
   1\\
   1\\
   1\\
   1\\
   1\\
  \end{bmatrix}  \quad \quad
{}^{3}\mathbf{G}^1=
  \begin{bmatrix}
    1 & 0 & 0 & 0\\
    1 & 1 & 0 & 0\\
    1 & 0 & 1 & 0\\
    1 & 1 & 1 & 1\\
  \end{bmatrix}  \quad \quad
{{}^{3}\mathbf{H}^1}^T=
  \begin{bmatrix}
   1\\
   1\\
   1\\
   1\\
  \end{bmatrix}
\]
\fi
\ifdefined\TwoCol
\[
{}^1\mathbf{G}^4=
  \begin{bmatrix}
    0 & 1 & 1 & 1 & 1\\
    1 & 0 & 1 & 1 & 1\\
    1 & 1 & 0 & 1 & 1\\
    1 & 1 & 1 & 0 & 1\\
    1 & 1 & 1 & 1 & 1\\
  \end{bmatrix},  \hspace{0.4em}
{{}^1\mathbf{H}^4}^T=
  \begin{bmatrix}
   1 & 0 & 0 & 0\\
   0 & 1 & 0 & 0\\
   0 & 0 & 1 & 0\\
   0 & 0 & 0 & 1\\
   1 & 1 & 1 & 1\\
  \end{bmatrix},
  \]
  \[
{}^2\mathbf{G}^3=
  \begin{bmatrix}
    1 & 1 & 1 & 1 & 1\\
    1 & 0 & 1 & 1 & 1\\
    1 & 1 & 0 & 1 & 1\\
    1 & 1 & 1 & 0 & 1\\
    1 & 1 & 1 & 1 & 0\\
  \end{bmatrix},  \hspace{0.4em}
{{}^2\mathbf{H}^3}^T=
  \begin{bmatrix}
   1 & 1 & 1\\
   1 & 1 & 0\\
   1 & 0 & 1\\
   1 & 0 & 0\\
   1 & 0 & 0\\
  \end{bmatrix}, 
\]
\[
{}^{4}\mathbf{G}^1=
  \begin{bmatrix}
    1 & 1 & 1 & 1 & 1\\
    1 & 0 & 1 & 1 & 1\\
    1 & 1 & 0 & 1 & 1\\
    1 & 1 & 1 & 0 & 1\\
    1 & 1 & 1 & 1 & 0\\
  \end{bmatrix},  \quad \quad
{{}^{4}\mathbf{H}^1}^T=
  \begin{bmatrix}
   1\\
   1\\
   1\\
   1\\
   1\\
  \end{bmatrix},
  \]
  \[
{}^{3}\mathbf{G}^1=
  \begin{bmatrix}
    1 & 0 & 0 & 0\\
    1 & 1 & 0 & 0\\
    1 & 0 & 1 & 0\\
    1 & 1 & 1 & 1\\
  \end{bmatrix}, \quad \quad
{{}^{3}\mathbf{H}^1}^T=
  \begin{bmatrix}
   1\\
   1\\
   1\\
   1\\
  \end{bmatrix}. \numberthis \label{MatricesExamples}
\]
\fi
All of these patterns and examples of $\mathbf{G}$ and $\mathbf{H}^T$ result in
\begin{equation} \label{ensure}
 \mathbf{G} \mathbf{H}^T=\begin{bmatrix}
    \mathbf{I}_k\\
    \mathbf{0}_{l \times k}\\
  \end{bmatrix},
\end{equation}
which demonstrates the validity of our matrices for wiretap coding. 
We provide another approach to obtain bins of our Ni code for all forms and cases in general, in the next section.
\section{Coding Technique of $\mathcal{B}\mathcal{Q}+\mathcal{S}$ } \label{The search for $AC+B$}
As mentioned in the previous section, the matrices $\mathbf{G}$ and $\mathbf{H}$ for odd cases do not exist (with the two exceptions, all forms of the case $l=1$ and the specific form $(3,1)$), because they are not linear. However, we discovered an alternative technique, $\mathcal{B}\mathcal{Q}+\mathcal{S}$, that can be used by the system to generate the code table for all forms of all cases. RASBA and RAHBA are two techniques that have been previously presented for generating bins from one form to another one. In this section, first, we will linearize the procedure of each of these techniques by using $\mathcal{B}\mathcal{Q}+\mathcal{S}$ technique, and then, we propose the combination approach to using both of them.

\subsection{$\mathbf{AC+B}$ for RASBA} \label{$AC+B$ for RASBA}

By employing the RASBA algorithm, the system can generate the next bins of the next form as follows 
\begin{align}  \label{^l_iC^k+1}
{}^{l}_{i}\mathcal{B}^{k+1}=
  \begin{bmatrix}
    {}^{l}_{i}\mathcal{B}^k\\
  \end{bmatrix} 
 \times   \begin{bmatrix}
    \mathbf{I}_{n} & \mathbf{0}_{n \times 1}\\
  \end{bmatrix}  + \begin{bmatrix} \begin{matrix} \mathbf{0}_{2^{l} \times n} \end{matrix} \quad \begin{matrix}
  \mathbf{[01]}^0_{2^{l} \times 1} \\
\end{matrix}   
    \end{bmatrix},
\end{align} 
\begin{align}     \label{^l_i+1C^k+1}
{}^{l}_{i+1}\mathcal{B}^{k+1}=
  \begin{bmatrix}
    {}^{l}_{i}\mathcal{B}^k\\
  \end{bmatrix} 
 \times   \begin{bmatrix}
    \mathbf{I}_{n} & \mathbf{0}_{n \times 1}\\
  \end{bmatrix}  + \begin{bmatrix} \begin{matrix} \mathbf{0}_{2^{l} \times n} \end{matrix} \quad \begin{matrix}
  \mathbf{[01]}^1_{2^{l} \times 1} \\
\end{matrix}   
    \end{bmatrix}, 
\end{align}
where ${}^{l}_{i}\mathcal{B}_{2^l \times n}^k$ is the $i$th bin of the form $\mathcal{D}_{l,k}$. In other words, it represents the $i$th bin in the code table of the form $k$ and the case $l$. 
In (\ref{^l_iC^k+1}) and (\ref{^l_i+1C^k+1}), we removed its sub-script to avoid repetition and make it simple to follow. Also, suppose a vertical sequence of the length $n$, consisting of $0$ and $1$ successively after each other, where it may start with either $0$ or $1$. If $\rho$ sequences are placed together as columns of a matrix, this sub-matrix is shown as $\mathbf{[01]}^i_{n \times \rho}$, where $i \in \{0,1,2\}$ and indicates the sequences start with $i=0$ or $i=1$. When $i=2$, it indicates that the general presentation that is symbolic and can be both. Thus, the total number of possible states for sub-matrix $\mathbf{[01]}^2_{n \times \rho}$ is $2^{\rho}$. Besides (\ref{^l_iC^k+1}) and (\ref{^l_i+1C^k+1}), the next $2^{k+1}-2$ bins of the form $(l,k+1)$ are made with the other $2^{k}-1$ bins of the form $(l,k)$.

In the next step, the recursion should be removed. The recursive model of generating bins by RASBA is as follows
\begin{equation} \label{removerec2}
\Bigg( \Big( ( \begin{bmatrix} {}^l_{i}\mathcal{B}^1 \\
\end{bmatrix}   \times {}^{l}\mathcal{Q}^2 + {}^{l}\mathcal{S}^2) \times {}^{l}\mathcal{Q}^3 + {}^{l}\mathcal{S}^3 \Big) \times {}^{l}\mathcal{Q}^4 + {}^{l}\mathcal{S}^4 \Bigg)...  
\end{equation}
where ${}^{l}\mathcal{Q}^k_{(n-1) \times n}$ and ${}^{l}\mathcal{S}^k_{2^l \times n}$ are matrices implied to attain the code table for the specific $l$ and $k$. We also deleted the right sub-scripts of matrices $\mathcal{Q}$ and $\mathcal{S}$ in (\ref{removerec2}) to minimize repetition. Furthermore, the $i$ which can be $1$ or $2$ in ${}^l_{i}\mathcal{B}^1$, indicates that the start point should be from the two bins of the $\mathcal{D}_{l,1}$. Since the procedure of generating bins is occurring in a specific case and every bin of the form $(l,k)$ results in two bins of the form $(l,k+1)$. Each bin of the form $(l,k+1)$ is directly obtained solely by one bin of the form $(l,k)$. As a result, (\ref{removerec2}) is an accurate algorithm. Let us start by considering the FF of the case $l=1$. Here the code table of $\mathcal{D}_{1,1}$, consists of the two known bins as follows:
\begin{equation} \label{C_111andC112}
{}^1_{1}\mathcal{B}^1=
\begin{bmatrix}
0 & 0\\
1 & 1\\
\end{bmatrix}, \quad \quad
{}^1_{2}\mathcal{B}^1= \begin{bmatrix}
0 & 1\\
1 & 0\\
\end{bmatrix}.
\end{equation}
By applying (\ref{removerec2}) to each of these bins, and following the steps, the result for $k=2$ would be
\begin{equation} \label{C_21i}
{}^{1}_{i}\mathcal{B}^{2}=
  \begin{bmatrix}
   {}^{1}_{j}\mathcal{B}^{1}
  \end{bmatrix} 
 \times   \begin{bmatrix}
    1 & 0 & 0\\
    0 & 1 & 0\\
  \end{bmatrix}  + \begin{bmatrix}  \begin{matrix}
  0 & 0\\
  0 & 0\\ 
    \end{matrix} \quad \begin{matrix}
    \mathbf{[01]}^2_{2 \times 1}\\
    \end{matrix}
    \end{bmatrix},
\end{equation}
where $\mathbf{[01]}^2_{2 \times 1}$ varies in $2$ states, and $j \in \{1,2\}$, so the permutations make all $4$ bins of the code table  (${}^{1}_{i}\mathcal{B}^{2}$ when $i\in \{1,2,3,4\}$).
For the case when $k=3$, there is
\begin{equation} \label{C_31i}
{}^{1}_{i}\mathcal{B}^{3}=
  \begin{bmatrix}
   {}^{1}_{j}\mathcal{B}^{1}
  \end{bmatrix} 
 \times   \begin{bmatrix}
    1 & 0 & 0 & 0\\
    0 & 1 & 0 & 0\\
  \end{bmatrix}  + \begin{bmatrix}  \begin{matrix}
  0 & 0 & 0\\
  0 & 0 & 0\\ 
    \end{matrix} \quad \begin{matrix}
    \mathbf{[01]}^2_{2 \times 2}\\
    \end{matrix}
    \end{bmatrix},
  \end{equation}
and again, since $\mathbf{[01]}^2_{2 \times 2}$ can produce $4$ states, when $j \in \{1,2\}$, the system can generate all $8$ bins of the code table.

For the $\mathcal{D}_{2,1}$, it starts from the two bins of the Ni code table
\begin{equation} \label{C_121andC122}
 {}^2_1\mathcal{B}^1=
\begin{bmatrix}
0 & 0 & 0\\
1 & 1 & 0\\
0 & 1 & 1\\
1 & 0 & 1\\
\end{bmatrix}, \quad \quad
{}^2_2\mathcal{B}^1= \begin{bmatrix}
0 & 0 & 1\\
1 & 1 & 1\\
0 & 1 & 0\\
1 & 0 & 0\\
\end{bmatrix},
\end{equation}
and all $4$ bins of the $\mathcal{D}_{2,2}$ can be achieved by
\begin{equation} \label{C_22i}
{}^2_i\mathcal{B}^2=
  \begin{bmatrix}
   {}^2_j\mathcal{B}^1
  \end{bmatrix} 
 \times   \begin{bmatrix}
    1 & 0 & 0 & 0\\
    0 & 1 & 0 & 0\\
    0 & 0 & 1 & 0\\
  \end{bmatrix}  + \begin{bmatrix}  \begin{matrix}
  0 & 0 & 0\\
  0 & 0 & 0\\ 
  0 & 0 & 0\\
  0 & 0 & 0\\ 
    \end{matrix} \quad \begin{matrix}
    \mathbf{[01]}^2_{4 \times 1}\\
    \end{matrix}
    \end{bmatrix},
  \end{equation}
with all permutations of $j \in \{1,2\}$ and $\mathbf{[01]}^2_{4 \times 1}$.
Also, all $8$ bins of the $\mathcal{D}_{2,3}$ can be attained as follows
\begin{equation} \label{C_32i}
 {}^2_i\mathcal{B}^3=
  \begin{bmatrix}
   {}^2_j\mathcal{B}^1
  \end{bmatrix} 
  \times   \begin{bmatrix}
    1 & 0 & 0 & 0 & 0\\
    0 & 1 & 0 & 0 & 0\\
    0 & 0 & 1 & 0 & 0\\
  \end{bmatrix}   + \begin{bmatrix}  \begin{matrix}
  0 & 0 & 0\\
   0 & 0 & 0\\ 
  0 & 0 & 0\\
  0 & 0 & 0\\ 
    \end{matrix} \quad \begin{matrix}
    \mathbf{[01]}^2_{4 \times 2}\\
    \end{matrix}
    \end{bmatrix},
  \end{equation}
with all permutations of $j=1,2$ and $\mathbf{[01]}^2_{4 \times 2}$. 

To sum up the procedure in a general way, we can propose this approach
\ifdefined\TwoCol
\begin{align*} 
{}^l_i\mathcal{B}^k=&
  \begin{bmatrix}
    {}^l_j\mathcal{B}^1\\
  \end{bmatrix} 
 \times   \begin{bmatrix}
    \mathbf{I}_{l+1} & \mathbf{0}_{(l+1) \times (k-1)}
  \end{bmatrix} \\
 & + \begin{bmatrix}  
  \mathbf{0}_{2^l \times (l+1)} & \mathbf{[01]}_{2^l \times (k-1)}\\
  \end{bmatrix}.   \numberthis \label{genrec}
\end{align*}
\fi
\ifdefined\OneCol
\begin{equation} \label{genrec}
{}^l_i\mathcal{B}^k=
  \begin{bmatrix}
    {}^l_j\mathcal{B}^1\\
  \end{bmatrix} 
 \times   \begin{bmatrix}
    \mathbf{I}_{l+1} & \mathbf{0}_{(l+1) \times (k-1)}
  \end{bmatrix} 
  + \begin{bmatrix}  
  \mathbf{0}_{2^l \times (l+1)} & \mathbf{[01]}_{2^l \times (k-1)}\\
  \end{bmatrix}.   
\end{equation}
\fi
Thus, this non-recursive general formula for the RASBA algorithm generates all bins of a form in a case from the two bins of the FF of that case.
\subsection{$\mathbf{AC+B}$ for RAHBA} \label{$AC+B$ for RAHBA}
Now, let us evaluate the standard method (the path of the red line in Fig. \ref{Schememap}) for the RAHBA algorithm, where the system is able to generate $\mathcal{D}$ of the FF of the next case from the $\mathcal{D}$ of the FF of the current case. There are always two bins in the FF of any cases. The two bins together make each of the bins of the $\mathcal{D}_{l+1,1}$, by switching all $0$ and $1$ which are appended to the right side of the codewords of each of the bins. Let us start the formulation of the process by generating the two bins of the next case as follows, respectively
  \begin{equation*}  \label{^l+1_1C^1}  
  {}^{l+1}_{1}\mathcal{B}^{1}=
  \begin{bmatrix}
    {}^{l}_{1}\mathcal{B}^1\\
    {}^{l}_{2}\mathcal{B}^1\\
  \end{bmatrix} 
 \times   \begin{bmatrix}
    \mathbf{I}_{(l+1)} & \mathbf{0}_{(l+1) \times 1}\\
  \end{bmatrix}  + \begin{bmatrix} \begin{matrix} \mathbf{0}_{(2^{l+1}) \times (l+1)} \end{matrix} \quad \begin{matrix}
  \mathbf{0}_{2^{l} \times 1} \\
  \mathbf{1}_{2^{l} \times 1} \\
\end{matrix}   
    \end{bmatrix},
  \end{equation*}
\begin{equation}  \label{^l+1_2C^1}
    {}^{l+1}_{2}\mathcal{B}^{1}=
  \begin{bmatrix}
    {}^{l}_{1}\mathcal{B}^1\\
    {}^{l}_{2}\mathcal{B}^1\\
  \end{bmatrix} 
 \times   \begin{bmatrix}
    \mathbf{I}_{(l+1)} & \mathbf{0}_{(l+1) \times 1}\\
  \end{bmatrix}  + \begin{bmatrix} \begin{matrix} \mathbf{0}_{(2^{l+1}) \times (l+1)} \end{matrix} \quad \begin{matrix}
  \mathbf{1}_{2^{l} \times 1} \\
  \mathbf{0}_{2^{l} \times 1} \\
\end{matrix}   
    \end{bmatrix}, 
  \end{equation}
when $k=1$, which means the bins belong to the FF. 

The objective now is to remove recursion on the RAHBA algorithm. The following is the recursive model of RAHBA to make the bins
\begin{equation} \label{removerec1}
\Bigg( \Big( ( \begin{bmatrix} {}^0_{1}\mathcal{B}^1 \\
{}^0_{2}\mathcal{B}^1\\
\end{bmatrix}   \times {}^{1}\mathcal{Q}^1 + {}^{1}\mathcal{S}^1) \times {}^{2}\mathcal{Q}^1 + {}^{2}\mathcal{S}^1 \Big) \times {}^{3}\mathcal{Q}^1 + {}^{3}\mathcal{S}^1 \Bigg)...   
\end{equation}

In order to attain the $\mathcal{D}$ of the FF of each case, the system should start from the $\mathcal{D}$ of the FF of the case $l=0$ which is $\mathcal{D}_{0,1}$. 
Thus, in (\ref{removerec1}), ${}^0_{1}\mathcal{B}^1$ and ${}^0_{2}\mathcal{B}^1$ must be the scalar values of $0$ and $1$, respectively, where each of them is the only codeword in a bin. Note that, the innermost parenthesis of (\ref{removerec1}) makes one of the two bins of the $\mathcal{D}_{1,1}$, and both of the bins should be concatenated vertically before being multiplied by the ${}^{2}\mathcal{Q}^1$. 
These two bins are as follows
\begin{equation} \label{C111andC112H}
{}^1_{1}\mathcal{B}^1=
  \begin{bmatrix}
    0\\
    1\\
  \end{bmatrix} 
 \times   \begin{bmatrix}
    1 & 0\\
  \end{bmatrix}  + \begin{bmatrix}  
  0 & 0\\
  0 & 1  
    \end{bmatrix}, \
{}^1_{2}\mathcal{B}^1=
  \begin{bmatrix}
    0\\
    1\\
  \end{bmatrix} 
 \times   \begin{bmatrix}
    1 & 0\\
  \end{bmatrix}  + \begin{bmatrix}  
  0 & 1\\
  0 & 0\\  
    \end{bmatrix}. 
\end{equation}
When $k=2$, the two next bins would be
\ifdefined\OneCol
 \begin{equation}  \label{C121andC122H}
{}^2_{1}\mathcal{B}^1=
  \begin{bmatrix}
    0\\
    1\\
    0\\
    1\\
  \end{bmatrix} 
 \times   \begin{bmatrix}
    1 & 0 & 0\\
  \end{bmatrix}  + \begin{bmatrix}  
  0 & 0 & 0\\
  0 & 1 & 0\\
  0 & 1 & 1\\
  0 & 0 & 1\\  
    \end{bmatrix}, \qquad
{}^2_{2}\mathcal{B}^1=
  \begin{bmatrix}
    0\\
    1\\
    0\\
    1\\
  \end{bmatrix} 
 \times   \begin{bmatrix}
    1 & 0 & 0\\
  \end{bmatrix}  + \begin{bmatrix}  
  0 & 0 & 1\\
  0 & 1 & 1\\
  0 & 1 & 0\\
  0 & 0 & 0\\  
    \end{bmatrix}, 
 \end{equation}
 \fi
\ifdefined\TwoCol
 \begin{equation*}  
{}^2_{1}\mathcal{B}^1=
  \begin{bmatrix}
    0\\
    1\\
    0\\
    1\\
  \end{bmatrix} 
 \times   \begin{bmatrix}
    1 & 0 & 0\\
  \end{bmatrix}  + \begin{bmatrix}  
  0 & 0 & 0\\
  0 & 1 & 0\\
  0 & 1 & 1\\
  0 & 0 & 1\\  
    \end{bmatrix},
    \end{equation*}
    \begin{equation} \label{C121andC122H}
{}^2_{2}\mathcal{B}^1=
  \begin{bmatrix}
    0\\
    1\\
    0\\
    1\\
  \end{bmatrix} 
 \times   \begin{bmatrix}
    1 & 0 & 0\\
  \end{bmatrix}  + \begin{bmatrix}  
  0 & 0 & 1\\
  0 & 1 & 1\\
  0 & 1 & 0\\
  0 & 0 & 0\\  
    \end{bmatrix}, 
 \end{equation}
 \fi
and the two bins of the $\mathcal{D}_{1,3}$ equal to
\ifdefined\OneCol
   \begin{equation} \label{C131andC132H}
{}^3_{1}\mathcal{B}^1=
  \begin{bmatrix}
    0\\
    1\\
    0\\
    1\\
    0\\
    1\\
    0\\
    1\\
  \end{bmatrix} 
 \times   \begin{bmatrix}
    1 & 0 & 0 & 0\\
  \end{bmatrix}  + \begin{bmatrix}  
  0 & 0 & 0 & 0\\
  0 & 1 & 0 & 0\\
  0 & 1 & 1 & 0\\
  0 & 0 & 1 & 0\\
  0 & 0 & 1 & 1\\
  0 & 1 & 1 & 1\\
  0 & 1 & 0 & 1\\
  0 & 0 & 0 & 1\\ 
    \end{bmatrix}, \quad
{}^3_{2}\mathcal{B}^1=
  \begin{bmatrix}
    0\\
    1\\
    0\\
    1\\
    0\\
    1\\
    0\\
    1\\
  \end{bmatrix} 
 \times   \begin{bmatrix}
    1 & 0 & 0 & 0\\
  \end{bmatrix}  + \begin{bmatrix}  
  0 & 0 & 0 & 1\\
  0 & 1 & 0 & 1\\
  0 & 1 & 1 & 1\\
  0 & 0 & 1 & 1\\
  0 & 0 & 1 & 0\\
  0 & 1 & 1 & 0\\
  0 & 1 & 0 & 0\\
  0 & 0 & 0 & 0\\ 
    \end{bmatrix}. 
   \end{equation}
\fi
\ifdefined\TwoCol
   \begin{equation*} 
{}^3_{1}\mathcal{B}^1=
  \begin{bmatrix}
    0\\
    1\\
    0\\
    1\\
    0\\
    1\\
    0\\
    1\\
  \end{bmatrix} 
 \times   \begin{bmatrix}
    1 & 0 & 0 & 0\\
  \end{bmatrix}  + \begin{bmatrix}  
  0 & 0 & 0 & 0\\
  0 & 1 & 0 & 0\\
  0 & 1 & 1 & 0\\
  0 & 0 & 1 & 0\\
  0 & 0 & 1 & 1\\
  0 & 1 & 1 & 1\\
  0 & 1 & 0 & 1\\
  0 & 0 & 0 & 1\\ 
    \end{bmatrix},
    \end{equation*}
    \begin{equation}   \label{C131andC132H}
{}^3_{2}\mathcal{B}^1=
  \begin{bmatrix}
    0\\
    1\\
    0\\
    1\\
    0\\
    1\\
    0\\
    1\\
  \end{bmatrix} 
 \times   \begin{bmatrix}
    1 & 0 & 0 & 0\\
  \end{bmatrix}  + \begin{bmatrix}  
  0 & 0 & 0 & 1\\
  0 & 1 & 0 & 1\\
  0 & 1 & 1 & 1\\
  0 & 0 & 1 & 1\\
  0 & 0 & 1 & 0\\
  0 & 1 & 1 & 0\\
  0 & 1 & 0 & 0\\
  0 & 0 & 0 & 0\\ 
    \end{bmatrix}. 
   \end{equation}
\fi
In a conclusion, the general form is as follows
  \begin{align}  \label{genhalfrec}
  {}^l_{1}\mathcal{B}^1=
  \begin{bmatrix}
    \mathbf{[01]}^0_{2^{l} \times 1}\\
  \end{bmatrix} 
 \times   \begin{bmatrix}
    1 & \mathbf{0}_{1 \times l}\\
  \end{bmatrix}  + \begin{bmatrix}  
  \mathbf{T}\\
    \end{bmatrix},
  \end{align}
where the $\mathbf{T}$ equals to
  \begin{equation} \label{T_Rahba}
\mathbf{T}=
  \begin{bmatrix}
    \mathbf{0}^{2^{l} \times 1} & \mathbf{Gr}^{l}_{rl}\\
  \end{bmatrix}, 
   \end{equation}
and the $\mathbf{Gr}^{l}_{rl}$ is the Gray code made by $l$ bits, which has $2^{l}$ rows and $l$ columns. The subscript $rl$ means the right-side left flipped.

The ${}^l_{2}\mathcal{B}^1$ is
\begin{align}    \label{genhalfrec2}
{}^l_{2}\mathcal{B}^1=
  \begin{bmatrix}
    \mathbf{[01]}^0_{2^{l} \times 1}\\
  \end{bmatrix} 
 \times   \begin{bmatrix}
    1 & \mathbf{0}_{1 \times l}\\
  \end{bmatrix}  + \begin{bmatrix}  
  \mathbf{T}_{u}\\
    \end{bmatrix}, 
\end{align}
where the $\mathbf{T}_{u}$ is the up-side down flipped version of $\mathbf{T}$.
As a result, in (\ref{genhalfrec}) and (\ref{genhalfrec2}), there are non-recursive general formulas to make the two bins of the $\mathcal{D}$ of the FF of any specific case.
\subsection{Combination} \label{Combination}
So far, we discovered the two non-recursive formulas for RASBA and RAHBA algorithms. Two formulas (\ref{genhalfrec}) and (\ref{genhalfrec2}) allow the system to generate the two bins of the $\mathcal{D}_{l,1}$ directly. Then, it is feasible to construct the bins of other values of $k$ in the current case, by using these two bins of the FF in the (\ref{genrec}). Now, the system can combine these two approaches by substituting (\ref{genhalfrec}) and (\ref{genhalfrec2}) in the ${}^1_j\mathcal{B}^l$ of equation (\ref{genrec}). By substituting (\ref{genhalfrec}) in (\ref{genrec}), it can get
\ifdefined\OneCol
\begin{equation}  \label{Final00}
{}^k_i\mathcal{B}^l=
  \begin{bmatrix}
    \begin{bmatrix}
    \mathbf{[01]}^0_{2^{l} \times 1}\\
  \end{bmatrix} 
 \times   \begin{bmatrix}
    1 & \mathbf{0}_{1 \times l}\\
  \end{bmatrix}  + \begin{bmatrix}  
  \mathbf{T}\\
    \end{bmatrix}\\
  \end{bmatrix} 
 \times  \begin{bmatrix}
    \mathbf{I}_{l+1} & \mathbf{0}_{(l+1) \times (k-1)}\\
  \end{bmatrix} + \begin{bmatrix}  
  \mathbf{0}_{2^l \times (l+1)} & \mathbf{[01]}_{2^l \times (k-1)}\\
  \end{bmatrix},
  \end{equation}
\fi
\ifdefined\TwoCol
\begin{align*} 
{}^k_i\mathcal{B}^l&=
  \begin{bmatrix}
    \begin{bmatrix}
    \mathbf{[01]}^0_{2^{l} \times 1}\\
  \end{bmatrix} 
 \times   \begin{bmatrix}
    1 & \mathbf{0}_{1 \times l}\\
  \end{bmatrix}  + \begin{bmatrix}  
  \mathbf{T}\\
    \end{bmatrix}\\
  \end{bmatrix} \\
 &\times  \begin{bmatrix}
    \mathbf{I}_{l+1} & \mathbf{0}_{(l+1) \times (k-1)}\\
  \end{bmatrix}
  + \begin{bmatrix}  
  \mathbf{0}_{2^l \times (l+1)} & \mathbf{[01]}_{2^l \times (k-1)}\\
  \end{bmatrix}, \numberthis \label{Final00}
  \end{align*} 
\fi
and after multiplication, the result is

\ifdefined\OneCol
  \begin{equation} \label{Final0}
  {}^k_i\mathcal{B}^l 
     = \begin{bmatrix}
  \begin{bmatrix}
    \begin{bmatrix}
    \mathbf{[01]}^0_{2^{l} \times 1}\\
  \end{bmatrix} 
 \times   \begin{bmatrix}
    1 & \mathbf{0}_{1 \times l}\\
  \end{bmatrix}  + \begin{bmatrix}  
  \mathbf{T}\\
    \end{bmatrix}\\
  \end{bmatrix}
 \ \mathbf{0}_{2^l \times (k-1)}\end{bmatrix} + \begin{bmatrix}  
  \mathbf{0}_{2^l \times (l+1)} & \mathbf{[01]}_{2^l \times (k-1)}\\
  \end{bmatrix},
  \end{equation}
  \fi
\ifdefined\TwoCol
  \begin{align*} 
  {}^k_i\mathcal{B}^l 
     &= \begin{bmatrix}
  \begin{bmatrix}
    \begin{bmatrix}
    \mathbf{[01]}^0_{2^{l} \times 1}\\
  \end{bmatrix} 
 \times   \begin{bmatrix}
    1 & \mathbf{0}_{1 \times l}\\
  \end{bmatrix}  + \begin{bmatrix}  
  \mathbf{T}\\
    \end{bmatrix}\\
  \end{bmatrix}
 \ \mathbf{0}_{2^l \times (k-1)}\end{bmatrix}\\
 &+ \begin{bmatrix}  
  \mathbf{0}_{2^l \times (l+1)} & \mathbf{[01]}_{2^l \times (k-1)}\\
  \end{bmatrix}, \numberthis \label{Final0}
  \end{align*}
\fi
and if we multiply it again, the result would be
  \begin{equation}    \label{final1}
  {}^k_i\mathcal{B}^l 
  = \begin{bmatrix} \
  \begin{bmatrix}
    \begin{bmatrix}
    \mathbf{[01]}^0_{2^{l} \times 1}\\
  \end{bmatrix} 
 \times   \begin{bmatrix}
    1 & \mathbf{0}_{1 \times l}\\
  \end{bmatrix}  + \begin{bmatrix}  
  \mathbf{T}\\
    \end{bmatrix}\\ 
  \end{bmatrix} \
  \mathbf{[01]}_{2^l \times (k-1)} \end{bmatrix}. 
  \end{equation} 
  
This makes half of the bins of the $\mathcal{D}_{l,k}$, and by substituting (\ref{genhalfrec2}) in (\ref{genrec}), the system can generate the other half of the bins of the $\mathcal{D}_{l,k}$, similar to the same process as follows
  \begin{align}   \label{final2}
  {}^k_i\mathcal{B}^l 
  &= \begin{bmatrix} \
  \begin{bmatrix}
    \begin{bmatrix}
    \mathbf{[01]}^0_{2^{l} \times 1}\\
  \end{bmatrix} 
 \times   \begin{bmatrix}
    1 & \mathbf{0}_{1 \times l}\\
  \end{bmatrix}  + \begin{bmatrix}  
  \mathbf{T}_{u}\\
    \end{bmatrix}\\ 
  \end{bmatrix} \
  \mathbf{[01]}_{2^l \times (k-1)} \end{bmatrix}, 
  \end{align}
where the only difference between (\ref{final1}) and (\ref{final2})  is the $\mathbf{T}$ which is replaced by $\mathbf{T}_{u}$.

Therefore, it is possible to construct the closed-form and non-recursive formulas of generating bins for RAHBA and RASBA algorithms, separately. Each of these can be considered as a part of the bins construction procedure of the Ni coding scheme. This closed-form and non-recursive formula has been introduced to generate the bins of every form of the suggested Ni code at once.
\section{Conclusion} \label{Conclusion}
This paper investigated the equivocation and coding techniques in the binary symmetric wiretap channel. We demonstrated how maximizing entropy for a single eavesdropper observation can generally lead to the maximization of the total equivocation by modeling the optimization as a linear programming (LP) problem. In the first phase of this work, we found a new limit that is tighter than the finite blocklength limit derived from Wyner's secrecy capacity, but with the application to the finite blocklength regime by solving this LP problem. In the second phase, we proposed a secure coding scheme that outperforms random binning codes and achieves the new limit for the case when the value of overhead bits is one ($l=1$), and almost achieves the limit for the cases when the value of overhead bits is more than one bit ($l=2,3,4,...$). Furthermore, we presented the two recursive algorithms RASBA and RAHBA, which are the foundation of the proposed Ni coding scheme. We explored the patterns to produce the generator and the parity-check matrices for the linear codes, and we presented the non-recursive variant of the RASBA and RAHBA algorithms, to obtain all Ni codes.
These proposed Ni coding techniques produce codes with a variety of values for overhead and message bits. Future contributions can investigate this idea of the coding problem as an optimization problem for other versions of wiretap channel models by defining the proper objective function and its constraints.
\section{Appendix A: Number of Possible Points}
This problem can be compared to a scenario where we need to fill $e$ spaces with $n+1$ different colored balls. There are no restrictions on how many balls of a single color we can choose, which means we are free to choose all or none balls from one color, and the balls with the same color are identical and totally equal. This problem can be solved by taking $\delta=\min\{e,n+1\}$ steps. In the $i$th step, we obligate ourselves to choose all $e$ balls, exactly from $i$ colors of balls. Therefore, there are $\binom{n+1}{i}$ options at every step. In the first sub-step of the $i$th step, we choose one ball from each of those $i$ colors. This guarantees that we choose at least one ball from those $i$ colors of balls to fill the whole $e$ spaces. In the second sub-step of this $i$th step, we wish to fill the remaining spaces ($e-i$ spaces), with those $i$ colors. However, in contrast to the previous sub-step, there is no obligation to choose at least one ball of each color, and similar to the first problem, we are free to choose all or none from one color. Thus, it is identical to the first problem, as was previously stated, and the problem has been reduced to an inner problem. It is similar to the first problem with $e-i$ spaces and $i$ colors of balls. Hence, the problem can be solved by addressing the same inner problems, which can be represented as follows: 
\begin{equation}  \label{appnumber}
N=\sum \limits_{i=1}^{\delta} \binom{n+1}{i}( \sum \limits_{j=1}^{e-i}\binom{i}{j}( \sum \limits_{k=1}^{e-i-j} \binom{j}{k}( \sum\limits_{m=1}^{e-i-j-k} \binom{k}{m}( ... )))),
\end{equation}
with the condition that when the upper limit of the $\sum$ is equal to $0$, the inner parenthesis containing that $\sum$ equals to $1$. However, this is the method by which we simulated and formed all the potential rows needed for the optimization. 

Another formula that produces the same result, which we derived by substituting the inner parenthesis of (\ref{appnumber}) with $\binom{e-1}{i-1}$, is as follows:
\begin{equation}   \label{appnumber2}
\sum \limits_{i=1}^{\delta} \binom{n+1}{i} \binom{e-1}{i-1}.  
\end{equation}
\bibliographystyle{IEEEtran}
\bibliography{IEEEabrv,references}  
\end{document}